\documentclass[11pt,a4paper]{amsart}
\usepackage[utf8]{inputenc}
\usepackage[a4paper,margin=2.6cm]{geometry}
\usepackage{amsmath,amssymb,amsthm,mathtools,abraces}
\usepackage{stmaryrd}
\usepackage{color,xcolor}
\usepackage[hypertexnames=false,hyperfootnotes=false,colorlinks=true,linkcolor=blue,%
citecolor=purple,filecolor=magenta,urlcolor=cyan,unicode,linktocpage=true,pagebackref=false]{hyperref}
\usepackage{yhmath}
\usepackage{verbatim}
\usepackage{enumitem}   
\usepackage{comment}

\newcommand\be{\begin{equation}}
\newcommand\ee{\end{equation}}

\newcommand\p{\partial}
\DeclareMathOperator{\odd}{odd}

\newcommand\normordboson{ {\scriptstyle {{*}\atop{*}}} }
\newcommand\Tr{{\rm Tr}\,}
\newcommand\diag{{\rm diag}\,}

\DeclareMathOperator{\gl}{\mathfrak{gl}}
\def\sppan{{\rm span}\,}

\DeclareMathOperator{\BGW}{BGW}
\DeclareMathOperator{\Gr}{Gr}

\makeatletter
\@namedef{subjclassname@2020}{%
  \textup{2020} Mathematics Subject Classification}
\makeatother

\newtheorem{theorem}{Theorem}
\newtheorem{lemma}{Lemma}[section]
\newtheorem{proposition}[lemma]{Proposition}

\newtheorem{remark}{Remark}[section]

\newtheorem{conjecture}{Conjecture}[section]
\newtheorem*{theorem*}{Theorem}

\newtheorem{definition}{Definition}

\numberwithin{equation}{section}
\usepackage{filecontents}

\title[On higher BGW tau-functions]{On higher Br\'ezin--Gross--Witten tau-functions}

\author{Alexander Alexandrov}

\author{Saswati Dhara}

\address{Center for Geometry and Physics, Institute for Basic Science (IBS), Pohang 37673, Korea
}

\email{ {\tt alexandrovsash at gmail.com}}

\email{ {\tt saswatidhara@ibs.re.kr}}


\date{\today}

\begin{document}

\begin{abstract} 
In this paper, we consider the higher Br\'ezin--Gross--Witten tau-functions, given by the matrix integrals. For these tau-functions we construct the canonical Kac--Schwarz operators, quantum spectral curves, and $W^{(3)}$-constraints. For the simplest representative we construct the cut-and-join operators, which describe the algebraic version of the topological recursion. We also investigate a one-parametric generalization of the higher Br\'ezin--Gross--Witten tau-functions. 
\end{abstract}

\maketitle


\tableofcontents

\newpage 

\def\thefootnote{\arabic{footnote}}

\setcounter{equation}{0}

\section{Introduction}
Cohomological field theories \cite{KM} provide a universal description of a huge family of enumerative geometry invariants. 
The Chekhov--Eynard--Orantin topological recursion on the regular spectral curves \cite{EO,EO1} is closely related to the Givental--Teleman description of cohomological field theories and for many cases \cite{DOSS} it allows us to describe the generating functions in terms of certain structures on the spectral curve. A generalization of this construction for the irregular spectral curves with 
simple ramification points, associated with degenerate cohomological field theories, was developed by Chekhov and Norbury \cite{CN}. This generalization describes Givental's decomposition formula for the partition functions on all, possibly irregular, spectral curves which near ramification points are similar to the Airy curve 
\be
2x=y^2
\ee
or the Bessel curve 
\be\label{BCI}
2xy^2=1.
\ee

Basic ingredients of this construction are the Kontsevich--Witten \cite{Wit91,Kon} and Br\'ezin--Gross--Witten (BGW) \cite{Brezin:1980rk,Gross:1980he} tau-functions.
Because of the special role, played by these functions in the modern mathematical physics, these functions are very well studied and many properties are known. In particular, these functions are solutions of the KdV hierarchy, and satisfy the Virasoro constraints. These constraints can be solved with the cut-and-join operators \cite{KSS,KS2}. Moreover, these functions are related to the intersection theory on the moduli spaces of the Riemann surfaces \cite{Wit91,Kon,Norb,NorbS}. In many respects these two tau-functions are very similar to each other.

A general description of decomposition  of the non-semi-simple cohomological field theories is not available yet. However, it is expected that it is closely related to a version of the Chekhov--Eynard--Orantin topological recursion on the spectral curves with 
higher ramification points \cite{Borot} locally described by the equation
\be\label{RegM}
(m+1)x=y^{m+1}
\ee
for $m\in {\mathbb Z}_{>0}$.
Associated partition functions were intensively investigated in the framework of matrix models, integrable systems, and associated string models in the early 90's of the last century, see  
\cite{KS,DVV,KMMMZ1,KMMMZ2,FKN,AM,IZ,Kon,Wag} and references therein. Canonical Kac--Schwarz (KS) operators, associated to this family, were recently constructed in \cite{H3_2}.

In this paper we investigate the generating functions which we claim to be associated with the spectral curves
\be\label{IrregM}
(m+1)xy^{m+1}=1
\ee
for $m\in {\mathbb Z}_{>0}$.
The simplest representative of this family is the Bessel curve \eqref{BCI}.
These generating functions are given by the asymptotic expansion of  the higher BGW matrix models, introduced by Mironov, Morozov and Semenoff \cite{Mironov:1994mv}. They are tau-functions of the $(m+1)$-reduction of the  Kadomtsev--Petviashvili (KP) hierarchy.
For these tau-functions we construct a complete description of the corresponding KS algebra. In particular, we derive the canonical pair of the KS operators. One of them plays the role of the quantum spectral curve, and another is the raising operator which generate the point of the Sato Grassmannian. Here we see an essential difference with the generalized Kontsevich models, associated to the curves \eqref{RegM}. Namely, for the family \eqref{RegM} the canonical KS operators are polynomials, while for the higher BGW tau-functions one of them is given by an infinite series. 

Using these KS operators we construct the $W^{(3)}$-constraints satisfied by the higher BGW tau-functions. 
For $m=2$, these constraints completely specify the tau-function. We solve these constraints in terms of the cut-and-join operators. This solution provides an algebraic version of the topological recursion. 

The tau-functions, associated to the spectral curves \eqref{RegM}, have a very nice enumerative geometry interpretation. Namely, according to Witten's conjecture \cite{Wag}, proved by Faber, Shadrin, and Zvonkine \cite{FSZ}, 
they are the generating functions of $r$-spin intersection theory, given by the intersection numbers of Witten's and psi classes. In the same time, the original BGW model corresponding to the Bessel spectral curve \eqref{BCI} describes the intersection theory of Norbury's classes \cite{Norb, NorbS, CGG}.  The geometric interpretation for the  higher BGW tau-functions is addressed in the recent paper of Chidambaram,  Garcia-Failde, and Giacchetto \cite{CGG}.

We also consider the one-parametric deformation of the higher BGW tau-functions by a logarithmic term introduced in \cite{Mironov:1994mv}.
For this deformation we construct a complete set of the KS operators, the quantum spectral curve, the $W^{(3)}$-constraints and for the simplest case with $m=2$ we construct the cut-and-join description. For the case of the original BGW tau-function ($m=1$) this deformation  was investigated in \cite{KS2}, where it was called the generalized BGW tau-function, so we call the deformed family the generalized  higher BGW tau-functions.
This deformation is very natural and we believe that it has a natural geometric interpretation. However, this interpretation is unclear yet.

We expect that the tau-functions, investigated in this paper, will play an important role in the topological expansion/Givental decomposition for the irregular spectral curves/non-semi-simple degenerate cohomological field theories.

\subsection{Organization of the paper} In Section \ref{S2}, we provide some elements of the general theory of the KP hierarchy, including the Sato Grassmannian description, the KS algebra, and the algebra of symmetries. Section \ref{S3} is devoted to the BGW tau-function. In particular, in this section we construct the canonical pair of the KS operators for this tau-function. In Section  \ref{S4} we investigate a family of tau-functions where the BGW tau-function is the simplest representative. A one-parametric deformation of this family is considered in Section \ref{S5}. 

\subsection{Acknowledgments}
This work was supported by the Institute for Basic Science (IBS-R003-D1). The authors are grateful to Sergey Shadrin for
inspiring discussions.
We are grateful to Nitin Kumar Chidambaram, Elba Garcia-Failde, and Alessandro Giacchetto for discussions and for sharing their results
prior to publication.

\section{KP hierarchy and its symmetries}\label{S2}

Let us briefly summarize by now standard  description of the KP hierarchy in terms of the Sato Grassmannian and KS operators; for more detail see, e.g., \cite{Sato,SegW,KS,AM,mulase1994algebraic, JMD,H3_2}
 and references therein. 

\subsection{KP hierarchy and Sato Grassmannian}\label{S2.1}

The KP hierarchy was introduced by Sato \cite{Sato}. It
can be represented in terms of a tau-function $\tau({\bf t})$ by the Hirota bilinear identity
\be\label{HBE}
\oint_{\infty} e^{\sum_{k>0} (t_k-t_k') z^k}
\tau ({\bf t}-[z^{-1}])\tau ({\bf t'}+[z^{-1}])dz =0,
\ee
which encodes all nonlinear equations of the KP hierarchy.
Here we use the standard short-hand notation
\be
{\bf t}\pm [z^{-1}]:= \bigl \{ t_1\pm   
z^{-1}, t_2\pm \frac{1}{2}z^{-2}, 
t_3 \pm \frac{1}{3}z^{-3}, \ldots \bigr \}.
\ee

Let us consider the description of the space of solutions for the KP hierarchy, introduced by Sato in \cite{Sato} and further developed by Segal and Wilson in \cite{SegW}. We work within the formal series setup, $\tau({\bf t})\in {\mathbb C}[\![t_1,t_2,t_3,\dots]\!]$. Hence,
we focus on Sato's version of the construction. Let us consider the space $H=H_+\oplus H_-$, where the subspaces 
\be
H_-=z^{-1}{\mathbb C}[\![z^{-1}]\!]
\ee and 
\be
H_+={\mathbb C}[z]
\ee 
are generated by negative and nonnegative powers of $z$ respectively. Then, the Sato Grassmannian $\rm{Gr}$ consists of all closed linear spaces $\mathcal{W}\in H $, which are compatible with $H_+$. Namely, an orthogonal projection $\pi_+ : \mathcal{W} \to H_+ $ should be a Fredholm operator, i.e. both the kernel ${\rm ker}\, \pi_+ \in \mathcal{W}$ and the 
cokernel ${\rm coker}\, \pi_+ \in H_+$ should be finite-dimensional vector spaces. 
The Grassmannian $\Gr$ consists of components $\Gr^{(k)}$, parametrized by an index of the operator $\pi_+$. We need only the component $\Gr^{(0)}$; other components have an equivalent description. 

Moreover, we will consider only the big cell $\Gr^{(0)}_+$ of $\Gr^{(0)}$, which is defined by the constraint ${\rm ker}\, \pi_+ = {\rm coker}\, \pi_+=0$. We call $\Gr^{(0)}_+$ the {\em Sato Grassmannian} for simplicity. There exists a bijection between the points of the Sato Grassmannian ${\mathcal W}\in\Gr^{(0)}_+$ and the tau-functions with $\tau({\bf 0})=1$.  Below for simplicity we consider only the tau-functions normalized by $\tau({\bf 0})=1$.

A point of the Sato Grassmannian ${\mathcal W}\in \rm{Gr}^{(0)}_+$ can be described by an  {\emph {admissible basis}} $\{\Phi_1^{\mathcal W},\Phi_2^{\mathcal W},\Phi_3^{\mathcal W},\dots\}$,
\be
{\mathcal W}=\sppan \{\Phi_1^{\mathcal W},\Phi_2^{\mathcal W},\Phi_3^{\mathcal W},\dots\}.
\ee
The crucial property of the admissible bases is that if $\{\Phi_j^{\mathcal W}\}$ and $\{ {\Phi'}_j ^{\mathcal W}\}$ are two admissible bases of ${\mathcal W}$, then the matrix which relates them is of the kind that has a determinant, or, equivalently, this matrix differs from the identity by an operator of trace class \cite{SegW}:

\begin{definition}\label{Defadm} 
 $\{\Phi_j^{\mathcal W}\}$ is an admissible basis for ${\mathcal W}\in \rm{Gr}^{(0)}_+$,  if
\begin{enumerate}
\item the linear map $H_+ \rightarrow H$ which takes $z^{j-1}$ to $\Phi_j^{\mathcal W}$ is injective and has image $\mathcal W$, and
\item the matrix, relating $\pi_+ (\Phi_j^{\mathcal W})$ to $z^{j-1}$ differs from the identity by an operator of trace class.  
\end{enumerate}
\end{definition}

We call an element of $H$ {\em monic} if its leading coefficient is equal to 1. Any point of the Sato Grassmannian has an admissible basis of the monic elements of the form
\be\label{goodbas}
\Phi_j^{\mathcal W}=z^{j-1}\left(1+O(z^{-1})\right),
\ee
of course, such basis is not unique. 

Let $M\in {\mathbb Z}_{\geq 1}$ and $\Lambda =\diag (\lambda_1,\dots,\lambda_M)$ be a diagonal matrix.
For any function $f({\bf t})$, dependent on the infinite set of variables ${\bf t}=(t_1,t_2,t_3,\dots)$, let
\be\label{Miwa}
f\left(\left[\Lambda^{-1}\right]\right):=f({\bf t})\Big|_{t_k=\frac{1}{k}\Tr \Lambda^{-k}}
\ee
be the {\em Miwa parametrization}. For any basis (\ref{goodbas}) the tau-function of the KP hierarchy in the Miwa parametrization is equal to the ratio of determinants
\be\label{miwatau}
\tau_{\mathcal W}([\Lambda^{-1}])=\frac{\det_{i,j=1}^M \Phi^{\mathcal W}_i(\lambda_j)}{\Delta(\lambda)},
\ee
where
\begin{equation}
\Delta(z)=\prod_{i<j}(z_j-z_i)
\end{equation}
is the {\em Vandermonde determinant}. 
Moreover, if for some function $\tau_{\mathcal W}$ equation (\ref{miwatau}) holds for all $M \in {\mathbb Z}_{\geq 1}$, then  $\tau_{\mathcal W}$ is a tau-function of the KP hierarchy.
 
 \subsection{Kac--Schwarz operators}
 Consider the ring of differential operators with coefficients formal Laurent series in the variable $z^{-1}$,
\be
{\mathcal D}:={\mathbb C}(\!(z^{-1})\!)[\![\frac{\p}{\p z}]\!]
\ee
and its subrings ${\mathcal D}_\pm:=H_\pm[\![\frac{\p}{\p z}]\!]$. Operators from ${\mathcal D}$ we will call differential operators. Below we allow the coefficients of operators in ${\mathcal D}$ to depend also on parameters, in particular, on $\hbar$.
A natural direct sum decomposition holds
\be
{\mathcal D}={\mathcal D}_+ \oplus{\mathcal D}_-.
\ee
\begin{definition}
For any point of the Sato Grassmannian ${\mathcal W}$ the {\em KS algebra}
\be
{\mathcal A}_{\mathcal W}:=\left\{{\mathtt a}\in {\mathcal D}\,\big | \,{\mathtt a}\cdot {\mathcal W} \subset {\mathcal W} \right\}
\ee
 is the algebra of the differential operators which stabilize this point.
\end{definition}
For a point of the Sato Grassmannian the KS  algebra is isomorphic to $w_{1+\infty}$ algebra \cite{H3_2}. The KS operators provide an efficient tool for investigation of tau-functions \cite{KS,KMMMZ1,FKN,AM,Mironov:1994mv,Alexandrov:2014cwa,H3_2}. For instance,  sometimes they allow to construct a complete set of linear Virasoro/W-constraints satisfied by the tau-functions. Moreover, they help to describe matrix models associated with tau-functions.
They are also closely related to the quantum spectral curves and relations between the Lax and the Orlov--Schulman operators. 

However, there is no universal way to construct the KS algebra ${\mathcal A}_{\mathcal W}$ explicitly for a given point of the Sato Grassmannian. For the tau-functions related to the enumerative geometry problems such algebras are often generated by relatively simple operators. In this paper we construct such a description for an interesting family of tau-functions.

Consider the space of the pairs of differential operators
\be\label{GrD}
\Gr_{\mathcal D}:= \left\{\left({\mathtt P},{\mathtt Q}\right)\in {\mathcal D}^2 \Big| \left[{\mathtt P},{\mathtt Q}\right]=1, {\mathtt P} \in\frac{\p}{\p z}+ z^{-1}  {\mathcal D}_-,{\mathtt Q}\in z+ {\mathcal D}_-\right\}.
\ee
In \cite{H3_2} it is proven that $\Gr_{\mathcal D}$ is isomorphic to the Sato Grassmannian $\Gr^{(0)}_+$. For any point of the Sato Grassmannian there is a unique pair of KS operators ${\mathtt P}_{\mathcal W}, {\mathtt Q}_{\mathcal W}\in{\mathcal A}_{\mathcal W}$ such that $({\mathtt P}_{\mathcal W}, {\mathtt Q}_{\mathcal W}) \in\Gr_{\mathcal D}$. This pair can be constructed explicitly for a given Sato group element. 
\begin{definition}
For a point ${\mathcal W}$ of the Sato Grassmannian  we call the pair of KS operators $({\mathtt P}_{\mathcal W}, {\mathtt Q}_{\mathcal W}) \in\Gr_{\mathcal D}$  the {\em canonical KS operators}, associated to ${\mathcal W}$ or the corresponding tau-function.
\end{definition}

In \cite[Corollary 2.2]{H3_2} it was shown that for any ${\mathcal W} \in  \Gr^{(0)}_+$ there are no non-trivial KS operators in ${\mathcal D}_-$, ${\mathcal A}_{\mathcal W} \cap {\mathcal D}_- =0$. Therefore, to find the canonical pair of KS operators it is sufficient to find any pair of KS operators, one from $z+{\mathcal D}_-$ and another from $\frac{\p}{\p z}+{\mathcal D}_-$, by uniqueness this pair is canonical.

The canonical pair of KS operators describes the point of the Sato Grassmannian in the following simple way. The first basis vector satisfies
\be
{\mathtt P}_{\mathcal W}\cdot \Phi_1^{\mathcal W}=0.
\ee 
This equation can be interpreted as the quantum spectral curve equation, it has a unique monic solution in $H$. Then, for this solution the vectors ${\mathtt Q}_{\mathcal W}^{j-1}\cdot \Phi_1^{\mathcal W}= z^{j-1}\left(1+O(z^{-1})\right)$ generate an admissible basis
\be\label{highQ}
{\mathcal W}=\sppan \{\Phi_1^{\mathcal W},{\mathtt Q}_{\mathcal W}\cdot \Phi_1^{\mathcal W},{\mathtt Q}_{\mathcal W}^2\cdot \Phi_1^{\mathcal W},\dots\}.
\ee

\subsection{Symmetries of the KP hierarchy}\label{S2.3}
Using the KS operators one can construct linear equations for the tau-functions. These equations are described by operators, which belong to a certain representation of the central extension of the algebra $\gl(\infty)$, for more details see, e.g., \cite{KS,FKN,Alexandrov:2014cwa,H3_2}. 
Let us remind the reader a description of the $W^{(3)}$ symmetry algebra of the KP hierarchy. 
The {\em Heisenberg--Virasoro subalgebra} of $\gl(\infty)$ is generated by the operators
\be
\widehat{J}_k =
\begin{cases}
\displaystyle{\frac{\p}{\p t_k} \,\,\,\,\,\,\,\,\,\,\,\, \mathrm{for} \quad k>0},\\[2pt]
\displaystyle{0}\,\,\,\,\,\,\,\,\,\,\,\,\,\,\,\,\,\,\, \mathrm{for} \quad k=0,\\[2pt]
\displaystyle{-kt_{-k} \,\,\,\,\,\mathrm{for} \quad k<0,}
\end{cases}
\ee
unit, and
\begin{equation}
\begin{split}
\label{virfull}
\widehat{L}_m&=\frac{1}{2} \sum_{a+b=m} \normordboson \widehat{J}_a \widehat{J}_b\normordboson\\
&=\frac{1}{2} \sum_{a+b=-m}a b t_a t_b+ \sum_{k=1}^\infty k t_k \frac{\p}{\p t_{k+m}}+\frac{1}{2} \sum_{a+b=m} \frac{\p^2}{\p t_a \p t_b}.
\end{split}
\end{equation}
Here the normal ordering $\normordboson\dots\normordboson$ puts all $\frac{\p}{\p t_k}$ to the right of all $t_k$, and we always assume that $t_k=0$ and $\frac{\p}{\p t_k}=0$ for negative $k$.
To extend the Heisenberg--Virasoro algebra to 
$W^{(3)}$ algebra it is necessary to consider also operators cubic in $\widehat{J}_k$,
\begin{multline}
\widehat{M}_k=\frac{1}{3} \sum_{a+b+c=k} \normordboson \widehat{J}_a \widehat{J}_b \widehat{J}_c \normordboson\\
=
\frac{1}{3}\sum_{a+b+c=-k}a\, b\, c\, t_a\, t_b\, t_c+\sum_{c-a-b=k}a\, b\, t_a\, t_b\, \frac{\p}{\p t_{c}}
+\sum_{b+c-a=k}a\, t_{a}\frac{\p^2}{\p t_b\p t_c}
+\frac{1}{3}\sum_{a+b+c=k}\frac{\p^3}{\p t_a \p t_b \p t_c}.
\end{multline}
The operators $\widehat{J}_k$, $\widehat{L}_k$, and $\widehat{M}_k$ satisfy the following commutation relations
\begin{equation}
\begin{split}\label{comm}
\left[\widehat{J}_k,\widehat{J}_m\right]&=k\, \delta_{k,-m},\\
\left[\widehat{J}_k,\widehat{L}_m\right]&=k\widehat{J}_{k+m},\\
\left[\widehat{L}_k,\widehat{L}_m\right]&=(k-m)\widehat{L}_{k+m}+\frac{1}{12}k(k^2-1)\delta_{k,-m},\\
\left[\widehat{J}_k,\widehat{M}_m\right]&=2k\,\widehat{L}_{k+m},\\
\left[\widehat{L}_k,\widehat{M}_m\right]&=(2k-m)\widehat{M}_{k+m}+\frac{1}{6}k(k^2-1)\widehat{J}_{k+m},
\end{split}
\end{equation}
and generate the $W^{(3)}$-algebra. 
A commutator of $\widehat{M}_k$'s contains the terms of fourth power in the current components $\widehat{J}_m$, so it can not be represented as a linear combination of $\widehat{J}_k$, $\widehat{L}_k$, and $\widehat{M}_k$.

To describe the action of the algebra (\ref{comm}) on the Sato Grassmannian let us introduce
\begin{equation}
\begin{split}\label{virw}
{\mathtt j}_k&=z^k,\\
{\mathtt l}_k&=-z^k\left(z\frac{\p}{\p z}+\frac{k+1}{2}\right),\\
{\mathtt m}_k&=z^k\left(\left(z\frac{\p}{\p z}\right)^2+(k+1)z\frac{\p}{\p z}+\frac{(k+1)(k+2)}{6}\right).
\end{split}
\end{equation}
These operators belong to ${\mathcal D}$ and satisfy the commutation relations (\ref{comm}) with the omitted central term. 

With any KS operator one can associate and describe explicitly a linear differential equation, satisfied by the tau-function.
The following lemma provides a particular instance of this general relation between the KS operators and linear constraints, satisfied by the tau-functions.
\begin{lemma}\label{lemma33}
If a point of the Sato Grassmannian $\mathcal{W}$ is stabilized by a KS operator
\be
\sum_k \left( a_k {\mathtt m}_k+ b_k {\mathtt l}_k+c_k {\mathtt j}_k\right ) \in {\mathcal A}_{\mathcal W},
\ee
where parameters $a_k$, $b_k$, and $c_k$ do not depend on ${\bf t}$,
then the corresponding tau-function $\tau_{\mathcal W}$ satisfies a linear constraint 
\be
\sum_k \left( a_k \widehat{M}_k+ b_k \widehat{L}_k+c_k \widehat{J}_k\right ) \cdot \tau_{\mathcal W} = \mu\, \tau_{\mathcal W}
\ee
for some eigenvalue $\mu$ independent of ${\bf t}$.
\end{lemma}

\section{BGW tau-function}\label{S3}

In this section, we remind the reader some basic properties of the BGW tau-function, for more detail see \cite{Mironov:1994mv,KS2}. Here we basically follow \cite{KS2}. We also derive a canonical pair of KS operators for the BGW tau-function. 

\subsection{BGW tau-function}

The BGW  model \cite{Brezin:1980rk,Gross:1980he}, also known as the Gross--Witten--Wadia model \cite{Wadia}, is given by the unitary matrix integral
\begin{equation}\label{umm}
Z_{\BGW}=\int_{U(M)}[dU]e^{\frac{1}{\hbar}\Tr\left(A^{\dagger}U+AU^{\dagger}\right)}.
\end{equation}
Here $M\in {\mathbb Z}_{>0}$, and $[dU]$ is the Haar measure on the unitary group $U(M)$ normalized by $\int_{U(M)} [dU] =1$;  $A$ and $A^{\dagger}$ are two external matrices. Using the invariance properties of the Haar measure it is easy to show that the partition function depends on their product, so it is convenient to introduce
\begin{equation}
\Lambda=(A^{\dagger}A)^{\frac{1}{2}}.
\end{equation}
For simplicity we assume that $\Lambda=\diag(\lambda_1,\dots,\lambda_M)$ is a positive  defined diagonal matrix and $\hbar$ is positive.

Let us focus on the expansion of the partition function for large values of the eigenvalues of the matrix $\Lambda$. To describe this expansion it is more convenient to use another matrix integral representation of the BGW partition function. Namely, according to Mironov--Morozov--Semenoff \cite{Mironov:1994mv} the partition function can be described by a particular generalized Kontsevich model
\begin{equation}\label{bgwp}
Z_{\BGW}=\frac{\displaystyle{\int [d \Phi] \exp\left(\frac{1}{\hbar}\text{Tr}\left(\frac{\Lambda^2 \Phi}{2}+\frac{1}{2 \Phi}-\hbar M\log  \Phi   \right)     \right)}}{\displaystyle{\int [d \Phi] \exp\left(\text{Tr}\left(\frac{1}{2\hbar \Phi}-M\log \Phi   \right)     \right)}}.
\end{equation}
The measure of integration here
\be\label{flatmeas}
[ d\Phi]:=\frac{1}{\prod_{k=1}^{M-1}k!} \prod_{i<j}d \Im \Phi_{ij} \,d \Re \Phi_{ij}\, \prod_{i=1}^M d \Phi_{ii}
\ee
is a flat measure associated with the space of $ M\times M$ Hermitian matrices.

In equation (\ref{bgwp}),  the integration is taken over $M\times M$ normal matrices $\Phi$, 
\begin{equation}\label{Phid}
\Phi=U ~\text{diag}(\varphi_1, \varphi_2,\ldots, \varphi_M)U^{\dagger},\quad \quad \varphi_i\in \gamma,
\end{equation}
where $U$ is a unitary matrix and the contour $\gamma$ (aka Hankel's loop) runs from $-\infty$ around the origin counter clockwise and returning to $-\infty$. Then, the measure of integration can be expressed in terms of $U$ and the eigenvalues $\varphi_i$ in the standard way
\begin{equation}
[d\Phi]=\Delta(\varphi)^2[dU]\prod_i^Md\varphi_i.
\end{equation}

Using the Harish-Chandra--Itzykson--Zuber formula we reduce the expression \eqref{bgwp} to the integral over the eigenvalues
\begin{equation}\label{BGWmod}
Z_{\BGW}=(-1)^{\frac{M(M-1)}{2}}\prod_1^M(j-1)!\frac{\text{det}_{i,j=1}^{M}\,\mathcal I_{M-i}(\lambda_j)}{\Delta\left(\frac{{\lambda}^2}{2\hbar}\right)},
\end{equation}
where
\begin{equation}\label{mbf1}
\mathcal I_{\nu}(z)=\frac{1}{2 \pi \iota}\int_{\gamma}e^{\frac{ z^2{\varphi}}{2\hbar}+\frac{1}{2\hbar{\varphi}}}\frac{d{\varphi}}{{\varphi}^{\nu+1}},
\end{equation}
and  $\iota=\sqrt{-1}$.
Function $\mathcal I_{\nu}(z)$ can be written as 
\begin{equation}
\mathcal I_{\nu}(z)=z^\nu I_{\nu}(z/\hbar),
\end{equation}
where $ I_{\nu}(z)$  is the modified Bessel function defined as
\begin{equation}\label{mbf}
 I_{\nu}(z)=\left(\frac{2}{z}\right)^{\nu}\frac{1}{2 \pi \iota}\int_{\gamma}e^{\frac{ z^2{\varphi}}{4}+\frac{1}{{\varphi}}}\frac{d{\varphi}}{{\varphi}^{\nu+1}}.
\end{equation}

Let $\tilde\lambda_j=\frac{\lambda_j}{\hbar}$, we consider a diagonal matrix, 
\begin{equation}
\tilde{\Lambda}=\Lambda/\hbar=\diag(\tilde\lambda_1,\tilde\lambda_2,\ldots,\tilde\lambda_M).
\end{equation}
We introduce
\begin{equation}
\tau_{\BGW}([\Lambda^{-1}])=\mathcal C^{-1}_{\BGW}Z_{\BGW},
\end{equation}
where
\begin{equation}
\mathcal C_{\BGW}=\frac{{ \Delta(\lambda)e^{\text{Tr}~\tilde{\Lambda}}\prod_{i=1}^{M}(j-1)!}}{{ \Delta\left(\frac{\lambda^2}{2\hbar}\right)\left(\frac{2\pi}{\hbar}\right)^{\frac{M}{2}}\prod_{i=1}^{M}\lambda_i^\frac{1}{2}}}.
\end{equation}
Then from \eqref{BGWmod} we have
\begin{equation}\label{detbgw}
\tau_{\BGW}([\Lambda^{-1}])=\frac{\text{det}_{i,j=1}^M\Phi_j^{\BGW}(\lambda_i)}{\Delta(\lambda)},
\end{equation}
where $\Phi_j^{\BGW}$ can be written in terms of the modified Bessel function 

\begin{equation}
\Phi_j^{\BGW}(\lambda)=\sqrt{2\pi \tilde{\lambda}}\lambda^{j-1}e^{-\tilde{\lambda}}I_{j-1}(\tilde{\lambda}).
\end{equation}
Here we consider the asymptotic expansion of $\Phi_j^{\BGW}(z)$ for large values of $|z|$, which can be obtained by the stationary phase method in integral \eqref{mbf}, 
\begin{equation}\label{phiexp}
\Phi_j^{\BGW}(z)=z^{j-1}\left(1+\sum_{k=1}^{\infty}\frac{(-\hbar)^k}{z^k}\frac{a_k(j)}{8^kk!}\right),
\end{equation}
 where $a_k(j)=(4(j-1)^2-1^2)(4(j-1)^2-3^2)\ldots (4(j-1)^2-(2k-1)^2)$. 
 
 Comparing  \eqref{phiexp} and \eqref{detbgw} with \eqref{goodbas} and \eqref{miwatau} respectively we conclude, that \eqref{detbgw} describes a tau-function of the 
KP hierarchy in the Miwa parametrization.  The vectors $\Phi_j^{\BGW}(z)$ define a point of the Sato Grassmannian, associated with the BGW tau-function
 \be
 {\mathcal W}_{\BGW}=\sppan \{\Phi_1^{\BGW}, \Phi_2^{\BGW},\dots \}\in \rm{Gr}^{(0)}_+.
 \ee
 Moreover, as we will see below, it describes the KdV reduction of the KP hierarchy. The KdV integrability of the BGW model was established by Mironov--Morozov--Semenoff \cite{Mironov:1994mv}.

\subsection{KS description of the BGW tau-function}
Operators
  \begin{equation}
\begin{split}\label{KSBGW}
  {\mathtt a}_{\BGW}={z}\frac{\partial}{\partial z}-\frac{1}{2}+\frac{z}{\hbar},\quad \quad
{\mathtt b}_{\BGW}=z^{2}
\end{split}
\end{equation}
stabilize the point $ {\mathcal W}_{\BGW}$ of the Sato Grassmannian \cite{Mironov:1994mv},
\be
  {\mathtt a}_{\BGW} \cdot {\mathcal W}_{\BGW} \in {\mathcal W}_{\BGW},\quad \quad  {\mathtt b}_{\BGW} \cdot {\mathcal W}_{\BGW} \in {\mathcal W}_{\BGW}.
\ee
Therefore, they are the KS operators for this point. These operators satisfy the commutation relation $\left[  {\mathtt a}_{\BGW},  {\mathtt b}_{\BGW}\right]=2{\mathtt b}_{\BGW}$.

Then, Lemma \ref{lemma33} allows one to construct linear constraints for the BGW tau-function. 
The KS operators ${\mathtt b}_{\BGW}^k$ describe a KdV reduction of the tau-function
\be\label{rc}
\frac{\p}{\p t_{2k}} \tau_{\BGW}=0.
\ee 
KS operators ${\mathtt b}_{\BGW}^k{\mathtt a}_{\BGW}$ imply the Virasoro constraints
\be\label{vc}
\left(\frac{1}{2}\widehat{L}_{2k}-\frac{1}{2\hbar}\frac{\p}{\p t_{2k+1}}+\frac{\delta_{k,0}}{16}\right)\cdot \tau_{\BGW}=0
\ee
for $k\geq 0$. These Virasoro constraints for the unitary matrix integral description of the BGW tau-function \eqref{umm} are derived by Gross and Newman \cite{GN}.

The BGW tau-function also satisfies the dimension constraint 
\begin{equation}\label{dimV}
 \hbar \frac{\p}{\p \hbar}  {\tau}_{\BGW}=\widehat{L}_0\cdot {\tau}_{\BGW}.
\end{equation}
Let us consider the operator 
\begin{equation}\label{CAJ}
\begin{split}
\widehat{W}_{\BGW}=\sum_{k,m\in {\mathbb Z}_{\odd}^+}^\infty \left(kmt_{k}t_{m}\frac{\p}{\p t_{k+m-1}}+\frac{1}{2}(k+m+1)t_{k+m+1}\frac{\p^2}{\p t_k \p t_m}\right)+\frac{t_1}{8}.
\end{split}
\end{equation}
From the constraints \eqref{rc}, \eqref{vc}, and \eqref{dimV}, satisfied by $\tau_{\BGW}$, it follows that the BGW tau-function satisfies the equation
\be
\frac{\p}{\p \hbar}\tau_{\BGW} =\widehat{W}_{\BGW}\cdot \tau_{\BGW}.
\ee
The solution normalized by $\tau_{\BGW}({\bf 0})=1$ is given by the following theorem. 
\begin{theorem}[\cite{KS2}]\label{TW}
\be\label{Wa}
\tau_{\BGW}=\exp\left({\hbar \widehat{W}_{\BGW}}\right)\cdot 1.
\ee
\end{theorem}
Let us consider the {\em topological expansion} of the tau-function,
\be
\tau_{\BGW}= \sum_{k=0}^\infty \tau_{\BGW}^{(k)}\hbar^k. 
\ee
The coefficients $\tau_{\BGW}^{(k)}$ are homogenous polynomials in ${\bf t}$.
They satisfy a linear recursion relation
\be\label{rec0}
 \tau_{\BGW}^{(k)} =\frac{\widehat{W}_{\BGW}}{k}\cdot  \tau_{\BGW}^{(k-1)}.
\ee
We call it the {\em algebraic topological recursion} to distinguish it from the Chekhov--Eynard--Orantin topological recursion.
\begin{remark}
According to Norbury's conjecture \cite{Norb}, recently proven in \cite{CGG}, the BGW tau-function is the generating function of certain intersection numbers on the moduli spaces of punctured Riemann surfaces. For this interpretation of the BGW tau-function the parameter $\hbar$ describes the expansion in the Euler characteristic of the punctured Riemann surface, and the operator $\widehat{W}_{\BGW}$ describes the change of topology. 
That is why slightly abusing notation we call such operators the {\em cut-and-join operators}. 
\end{remark}
The operator $\widehat{W}_{\BGW}$ consists of an infinite number of terms, however, on each step of the recursion \eqref{rec0} only a finite number of them contribute. Hence, the recursion is given by the action of the polynomial differential operators on polynomials.

\subsection{Canonical KS operators}\label{S3.3}

In  \cite{KS2} it was shown that the KS operators ${\mathtt a}_{\BGW}$ and 
  \begin{equation}
\begin{split}
 {\mathtt c}_{\BGW}&=\frac{1}{{\mathtt b}_{\BGW}}  {\mathtt a}_{\BGW}^2\\
 &=\frac{\partial^2}{\partial z^2}+\frac{2}{\hbar}\frac{\partial}{\partial z}+\frac{1}{\hbar^2}+\frac{1}{4z^2}
\end{split}
\end{equation}
completely specify the point ${\mathcal W}_{\BGW}$ of the Sato Grassmannian.
In particular, it means that the canonical KS operators ${\mathtt P}_{\BGW}$ and ${\mathtt Q}_{\BGW}$ can be expressed in terms of ${\mathtt a}_{\BGW}$ and ${\mathtt c}_{\BGW}$. 
Let us find the canonical pair of KS operators for the BGW tau-function.
\begin{proposition}\label{Prop3.1}
For the BGW tau-function the canonical pair of KS operators is given by
 \begin{equation}
\begin{split}\label{QQ}
{\mathtt P}_{\BGW}&=\sqrt{{\mathtt c}_{\BGW}} -\frac{1}{\hbar},\\
{\mathtt Q}_{\BGW}&=\left({\mathtt a}_{\BGW}+\frac{1}{2}\right)\frac{1}{\sqrt{{\mathtt c}_{\BGW}}}.
\end{split}
\end{equation}
 \end{proposition}
 
 \begin{proof} 
To prove the first equation we need to show that $\sqrt{{\mathtt c}_{\BGW}} -\frac{1}{\hbar}\in \frac{\p}{\p z}+{\mathcal D}_-$ and that it is a KS operator for $\mathcal{W}_{\BGW}$.
Consider the operator 
\be
{\mathtt c}_0=\left(\frac{\partial}{\partial z}+\frac{1}{\hbar}\right)^2.
\ee
This is not a KS operator for $\mathcal{W}_{\BGW}$. We have
$\sqrt{{\mathtt c}_0}-\frac{1}{\hbar}=\frac{\p}{\p z}$.
Since $ {\mathtt c}_{\BGW}-{\mathtt c}_0=\frac{1}{4z^2}$,  we have
\be\label{f1}
\sqrt{{\mathtt c}_{\BGW}}-\frac{1}{\hbar}= \frac{\p}{\p z} + {\mathtt d},
\ee
where ${\mathtt d} \in  z^{-1}{\mathcal D}_-$.
Moreover, 
\be
\sqrt{{\mathtt c}_{\BGW}}=\frac{1}{\hbar}\sqrt{1+{\mathtt  A}}=\frac{1}{\hbar}\left(1+\frac{{\mathtt  A}}{2}-\frac{{\mathtt  A}^2}{8}+\dots\right)
\ee
where
  \be
{\mathtt  A}=\hbar^2{\mathtt c}_{\BGW}-1= \hbar^2\frac{\partial^2}{\partial z^2}+2\hbar\frac{\partial}{\partial z}+\frac{\hbar^2}{4z^2}
  \ee
is a KS operator. Therefore, $\sqrt{{\mathtt c}_{\BGW}}$ is a KS operator, and from \eqref{f1} we have the first equation in \eqref{QQ}.

We also have
\be\label{canb}
\left({\mathtt a}_{\BGW}+\frac{1}{2}\right)\frac{1}{\sqrt{{\mathtt c}_{\BGW}}}=z\left(\frac{\partial}{\partial z}+\frac{1}{\hbar}\right)\frac{1}{\frac{\partial}{\partial z}+\frac{1}{\hbar}+{\mathtt d}}\in z+ {\mathcal D}_-.
\ee
The operator in \eqref{canb} is a KS operator, therefore it coincides with ${\mathtt Q}_{\BGW}$. This completes the proof.
\end{proof}
 
The canonical pair of KS operators completely specify the point of the Sato Grassmannian  \cite{H3_2}. In particular,
the operator ${\mathtt P}_{\mathcal W}$ annihilates the first basis vector, hence it defines the quantum spectral curve. For  $\mathcal{W}_{\BGW}$ we have
\be
{\mathtt P}_{\BGW}\cdot \Phi_1^{\BGW}=0,
\ee
or, equivalently,
\be\label{eqs}
\sqrt{{\mathtt c}_{\BGW}}\cdot\Phi_1^{\BGW} = \frac{1}{\hbar}\Phi_1^{\BGW}.
\ee
This equation has a unique monic solution in the space of formal Laurent series ${\mathbb C}(\!(z^{-1})\!)$. From \eqref{eqs} we have
\be
\left({\mathtt c}_{\BGW}- \frac{1}{\hbar^2}\right)\cdot\Phi_1^{\BGW} =0,
\ee 
or
\be
\left(\frac{\partial^2}{\partial z^2}+\frac{2}{\hbar}\frac{\partial}{\partial z}+\frac{1}{4z^2}\right)\cdot\Phi_1^{\BGW} =0.
\ee
This is the quantum Bessel curve equation  \cite{KS2,DN2}. All higher basis vectors can be obtained by the action of the raising operator ${\mathtt Q}_{\BGW}$, see \eqref{highQ}.
\begin{remark}
In \cite{H3_2} the canonical KS operators were constructed for the Kontsevich--Witten tau-function, a close cousin of the BGW tau-function. Here we see a difference in the KS description of two models. While the canonical KS operators for the Kontsevich--Witten tau-function are polynomial in $\frac{\p}{\p z}$, $z$ and $z^{-1}$, they are infinite sums in $z^{-1}$ and $\frac{\p}{\p z}$ for the BGW case.
\end{remark}

\section{Higher BGW tau-functions}\label{S4}

In this section, we investigate the higher BGW tau-functions, introduced by Mironov--Morozov--Semenoff \cite{Mironov:1994mv}.
In particular, we construct a complete set of the KS operators,  the quantum spectral curve, the $W^{(3)}$-constraints and  the cut-and-join description for the simplest case $m=2$.

\subsection{Matrix models and tau-functions}\label{S4.1}

Following Mironov--Morozov--Semenoff \cite{Mironov:1994mv} we consider the generalized Kontsevich models with antipolynomial (to be precise, antimonomial) potential
\be\label{Ebgwp}
Z^{(m)}=\frac{\displaystyle{\int [d \Phi] \exp\left(\frac{1}{\hbar}\text{Tr}\left(\frac{1}{m+1} \Lambda^{m+1} \Phi+\frac{1}{m(m+1)  \Phi^m}-\hbar M\log  \Phi   \right)     \right)}}{\displaystyle{\int [d \Phi] \exp\left(\frac{1}{\hbar}\text{Tr}~~\left(\frac{1}{m(m+1)  \Phi^m}-\hbar M\log \Phi   \right)     \right)}}. 
\ee
Below we call them the {\em higher BGW} models. Original BGW model \eqref{bgwp} corresponds to $m=1$, 
\be
Z_{\BGW}=Z^{(1)},
\ee
and below we assume that $m \in {\mathbb Z}_{>0}$. In \eqref{Ebgwp} we integrate over the $M\times M$ normal matrices $\Phi$ associated to the Hankel's loop $\gamma$, see \eqref{Phid}. For simplicity we again assume that $\Lambda=\diag(\lambda_1,\dots,\lambda_M)$ is a positive  defined diagonal matrix and $\hbar$ is positive. For arbitrary complex $\lambda_j$ and $\hbar$ the contours defining the normal matrix $\Phi$ would be rotated correspondingly.

After the integration over the unitary group with the help of the Harish-Chandra--Itzykson--Zuber formula the matrix model \eqref{Ebgwp} reduces to
\begin{equation}\label{detH}
Z^{(m)}=(-1)^{\frac{M(M-1)}{2}}\prod_{j=1}^M(j-1)!  \frac{\text{det}_{i,j=1}^{M}\left(\mathcal I_{M-i}^{(m)}(\lambda_j)  \right)}{\Delta\left({\frac{{\lambda}^{m+1}}{(m+1)\hbar}}\right)},
\end{equation}
where
\begin{equation}\label{bess}
 \mathcal I_{\nu}^{(m)}(z)=\frac{1}{2 \pi \iota}\int_{\gamma}e^{\frac{1}{m(m+1)\hbar}\left(mz^{m+1}\varphi+\frac{1}{\varphi^{m}}\right)}\frac{d \varphi}{\varphi^{\nu+1}}.
\end{equation}
Note that, $ \mathcal I_{\nu}^{(1)}(z)= \mathcal I_{\nu}(z)$ given by equation (\ref{mbf1}).

In this paper we are focused on the asymptotic expansion of the partition function $Z^{(m)}$ for the large eigenvalues $\lambda_i$. This allows us to associate a tau-function of the KP hierarchy with a higher BGW model.
With the quasi-classical prefactor
\be\label{qqpref}
{\mathcal C}_m=\frac{{ \Delta(\lambda)e^{\frac{\text{Tr} {\Lambda}^m}{m\hbar}}\prod_{i=1}^{M}(j-1)!}}{{ \Delta\left(\frac{\lambda^{m+1}}{(m+1)\hbar}\right) \left(\frac{2\pi}{\hbar}\right)^{\frac{M}{2}}\prod_{i=1}^{M}\lambda_i^\frac{m}{2}}},
\ee
we have the determinant formula
\be\label{qqpref1}
{\mathcal C}_m^{-1} Z^{(m)}=\frac{\text{det}_{i,j=1}^M\Phi_j^{(m)}(\lambda_i)}{\Delta(\lambda)}.
\ee
Here
\begin{equation}
\begin{split}\label{bv}
\Phi_j^{(m)}(z)&=\sqrt{\frac{2\pi {z}^m}{\hbar}}e^{-\frac{z^m}{m\hbar}}\mathcal I_{j-1}^{(m)}({z})\\
&=\sqrt{\frac{2\pi {z}^m}{\hbar}}e^{-\frac{z^m}{m\hbar}}\frac{1}{2 \pi \iota}\int_{\gamma}e^{\frac{1}{m(m+1)\hbar}\left(m z^{m+1}\varphi+\frac{1}{\varphi^{m}}\right)}\frac{d \varphi}{\varphi^{j}},
\end{split}
\end{equation}
and  in the last line we assume the asymptotic expansion of the integral at large positive $z$.

Let us consider this asymptotic expansion in more detail. For the potential
\be
V(\varphi)=\frac{1}{m(m+1)\hbar}\left(m z^{m+1}\varphi+\frac{1}{\varphi^{m}}\right)
\ee
we consider the critical points of the potential satisfying
\be
\frac{\p}{\p \varphi}V(\varphi)=0.
\ee
The solutions $\varphi_k=z^{-1} e^{\iota \frac{2 \pi k}{m+1}}$, $k\in \{0,1,\dots,m\}$ correspond to the roots of the unity. The contour of integration $\gamma$ can be smoothly deformed without crossing these points in a way that the new contour contains an arc $(z^{-1}-\iota \epsilon,z^{-1}+\iota \epsilon)$ with small $\epsilon$. Here $z$ is assumed to be positive. The steepest descent method allows us to find the asymptotic expansion of the integral by the expansion in the neighborhood of the stationary point $\varphi_0=z^{-1}$. This point corresponds to the dominant contribution in \eqref{bv}, contributions of the other critical points are suppressed and can be neglected. 
We have $V''(\varphi_0)>0$, therefore it is convenient to consider a change of variables $\varphi\mapsto \varphi_0+\iota \frac{\varphi}{V''(\varphi_0)}$. Then, the steepest descent method gives the asymptotic expansion
\begin{equation}
\mathcal I_j^{(m)}(z)=\frac{\varphi_0^{-j}e^{V(\varphi_0)}}{ 2\pi\sqrt{V''(\varphi_0)}}\int_{\mathbb R} d\varphi \left(1+\frac{\iota \varphi}{\varphi_0\sqrt{V''(\varphi_0)}}\right)^{-j}e^{-\frac{\varphi^2}{2}+\sum_{\ell=3}^{\infty}t_{\ell}^*\varphi^{\ell}}.
\end{equation} 
Here
\be
t_{\ell}^*=\frac{(\iota)^{\ell}}{\ell ! }\frac{V^{(\ell)}(\varphi_0)}{V''(\varphi_0)^{\frac{\ell}{2}}}. 
\ee
More explicitly, we have
\be
\mathcal I_j^{(m)}(z)=z^{j-1} e^{\frac{z^m}{m \hbar}}\frac{1}{2\pi} \sqrt{\frac{\hbar }{   z^{m}}}\int_{\mathbb R} d\varphi \left(1+ \iota \varphi  \sqrt{\frac{\hbar}{z^{m}}}\right)^{-j}e^{-\frac{\varphi^2}{2}+\sum_{\ell=3}^{\infty}t_{\ell}^*\varphi^{\ell}}
\ee
with
\be
t_{\ell}^*=\frac{(-\iota)^{\ell}}{\ell!}\frac{(m+\ell-1)!}{(m+1)!} \left(\frac{\hbar}{z^m}\right)^{\frac{\ell}{2}-1}.
\ee
and $t_{\ell}^*=0$ for $\ell<3$.
Using this asymptotic expansion of \eqref{bv} we introduce
\begin{definition}\label{def1}
\be
\Phi_j^{(m)}(z)=z^{j-1}\frac{1}{\sqrt{2\pi}}\int_{\mathbb R} d\varphi \left(1+ \iota \varphi  \sqrt{\frac{\hbar}{z^{m}}}\right)^{-j}e^{-\frac{\varphi^2}{2}+\sum_{\ell=3}^{\infty}t_{\ell}^*\varphi^{\ell}}.
\ee
\end{definition}
\begin{remark}\label{R4.1}
These formal series are well defined for all $j\in {\mathbb Z}$ with $\Phi_j^{(m)}(z)=z^{j-1}(1+O(z^{-1}))$.
Moreover, $z^{N}\Phi_{j-N}^{(m)}(z)\in z^{j-1}(1+O(z^{-1}))$ is a well defined formal series of $z^{-1}$ and $N$ for arbitrary $j\in {\mathbb Z}$.
\end{remark}
Consider the expansion
 $e^{\sum_{\ell=3}^{\infty}t_{\ell}^*\varphi^{\ell}}=1+ \sum_{k=3}^\infty p_{k}({\bf t^*})\varphi^k$, where $p_{k}$ is the elementary Schur function. From Definition \ref{def1} we have
\begin{multline}\label{eqex}
\Phi_j^{(m)}(z)=z^{j-1}\left(1+\sum_{k=1}^{\infty}\begin{pmatrix}
j+2k-1 \\
j-1
\end{pmatrix} \left(-\frac{ \hbar}{z^{m}}\right)^{k}(2k-1)!!\right.\\
\left.
+\sum_{\ell=2}^{\infty}p_{2\ell}({\bf t^*})(2\ell-1)!!+g_j(m,z)  \right),
\end{multline}
where
\begin{equation}
\begin{split}
g_j(m,z)&= 
   \sum_{\ell=1}^{\infty}p_{2\ell+1}({\bf t}^*)\sum_{k=0}^{\infty}\begin{pmatrix}
j+2k \\
j-1
\end{pmatrix} \left(\frac{-\iota \hbar^{\frac{1}{2}}}{z^{m/2}}\right)^{2k+1}(2\ell+2k+1)!!\\
&+   \sum_{\ell=2}^{\infty}p_{2\ell}({\bf t}^*)\sum_{k=0}^{\infty}\begin{pmatrix}
j+2k+1 \\
j-1
\end{pmatrix} \left(-\frac{ \hbar}{z^{m}}\right)^{k+1}(2\ell+2k+1)!!
\end{split}
\end{equation}
and $
\begin{pmatrix}
n \\
r
\end{pmatrix}=\frac{n!}{r!(n-r)!}$ are the binomial coefficients.

\subsection{Higher BGW model as a tau-function of KP hierarchy}

The prefactor \eqref{qqpref} is chosen in such a way that
\begin{equation}\label{norm}
\Phi_j^{(m)}(z)=z^{j-1}\left(1+O(z^{-m})\right).
\end{equation}
These $\Phi_j^{(m)}(z)$ constitute an admissible basis for a point in the Sato Grassmannian
\begin{equation}\label{Gp}
\mathcal{W}_{m}=\sppan\{ \Phi_1^{(m)},\Phi_2^{(m)},\Phi_3^{(m)},\ldots\}\in \rm{Gr}^{(0)}_+.
\end{equation}
Therefore, the higher BGW models 
\be\label{HBGW}
\tau^{(m)}\left(\left[\Lambda^{-1}\right]\right)={\mathcal C}_m^{-1} Z^{(m)}
\ee
define tau-functions of the KP hierarchy in the Miwa parametrization
\begin{equation}\label{deter}
\tau^{(m)}\left(\left[\Lambda^{-1}\right]\right)=\frac{\text{det}_{i,j=1}^M\Phi_j^{(m)}(\lambda_i)}{\Delta(\lambda)}.
\end{equation}
We call them the {\em higher BGW tau-functions}. For $m=1$ the tau-function coincides with the BGW tau-function considered in Section \ref{S3}, $\tau^{(1)}=\tau_{\BGW}$ and $\Phi_k^{(1)}=\Phi_k^{\BGW}$. From \eqref{norm} we have $\tau^{(m)}({\bf 0})=1$. Below we investigate these tau-functions.

Let us introduce the Euler operator
\be\label{Eul}
\widehat{D}=\sum_{k=1}^\infty k t_k \frac{\p}{\p t_k},
\ee
which coincides with $\widehat{L}_0$.
\begin{lemma}\label{lemma31}
The tau-function $\tau^{(m)}$ satisfies the homogeneity condition 
\be
\widehat{D}\cdot \tau^{(m)}=m  \hbar \frac{\p}{\p \hbar} \tau^{(m)}. 
\ee
\end{lemma}
\begin{proof}
Up to a simple prefactor, the functions $\Phi_j^{(m)}(z)$ depend only on a combination of $z$ and $\hbar$, namely $\frac{z^m}{\hbar}$.
Indeed, if we change the variable of integration $\varphi \mapsto \varphi/z$ in \eqref{bv}, we have
\be
\Phi_j^{(m)}(z)=z^{j-1} \sqrt{\frac{2\pi {z}^m}{\hbar}}e^{-\frac{z^m}{m\hbar}}\frac{1}{2 \pi \iota}\int_{\gamma}e^{\frac{z^{m}}{m(m+1)\hbar}\left(m\varphi+\frac{1}{\varphi^{m}}\right)}\frac{d \varphi}{\varphi^{j}}=z^{j-1}f_j\left(\frac{\hbar}{z^m}\right)
\ee
for some functions $f_j$. Then the statement of the lemma follows immediately from the determinant formula (\ref{deter}) and the definition of the Miwa parametrization \eqref{Miwa}.
\end{proof}

Using the expansion \eqref{eqex} it is easy to find explicitly the leading coefficients of the expansion of basis vectors 
\be
\Phi_j^{(m)}=z^{j-1}\left(1+\sum_{k=1}^\infty\Phi_{j,k}^{(m)} \left(\frac{\hbar}{z^k}\right)^m\right),
\ee
they are presented in Appendix \ref{A1}. They are polynomials, $\Phi_{j,k}^{(m)} \in {\mathbb Q}[j,m]$.
For $m=1$ the coefficients are given by \eqref{phiexp}.

Let us consider the perturbative expansion of the functions $ \mathcal I_{\nu}^{(m)}(z)$, given by the contour integrals \eqref{bess} for  small $|z|$.
We put 
\be
\chi=\frac{z^{m+1}}{(m+1)\hbar}.
\ee
Let us change the variable of integration $\varphi \mapsto \varphi/\chi$. Then
\begin{equation}
\begin{split}\label{smallzexp}
 \mathcal I_{\nu}^{(m)}(z)&=\chi^\nu \frac{1}{2 \pi \iota}\int_{\gamma}e^{\varphi+\frac{\chi^{m}}{m(m+1)\hbar}\frac{1}{\varphi^{m}}}\frac{d \varphi}{\varphi^{\nu+1}}\\
 &=\chi^\nu\sum_{k=0}^{\infty}\frac{1}{k!}\left(\frac{\chi^{m}}{m(m+1)\hbar}\right)^k\frac{1}{\Gamma(k m+\nu+1)}\\
 &= \frac{\chi^\nu}{\Gamma(\nu+1)}{}_0F_m \left(;\left\{\frac{1+\nu}{m},\frac{2+\nu}{m}\ldots,\frac{m+\nu}{m}\right\};\frac{\chi^m}{m^{m+1}(m+1)\hbar}\right),
\end{split}
\end{equation}
where we use the Hankel's loop integral formula for the reciprocal gamma function
\be
\frac{1}{\Gamma(\nu+1)}= \frac{1}{2 \pi \iota}\int_{\gamma}e^{\varphi}\frac{d \varphi}{\varphi^{\nu+1}}.
\ee
Here the generalized hypergeometric function ${}_0F_m \left(;\{b_1,b_2,\ldots,b_m\};x\right)$ is defined by
\begin{align}
{}_0F_m \left(;\{b_1,b_2,\ldots,b_q\};x\right)&=\sum_{k=0}^{\infty}\frac{1}{(b_1)_k(b_2)_k\ldots(b_m)_k}\frac{x^k}{k!},
\end{align}
where $(b)_k=\frac{\Gamma(b+k)}{\Gamma(b)}$. For $m=1$ this expression is closely related to the expansion of the modified Bessel function \eqref{mbf}. 

For large values of $|z|$ the coefficients of $\Phi_j$'s also coincide with the coefficients of the assymptotic expansion of the of hyper-geometric functions. We have compared the first three terms of these expressions for $m=1$ and $m=2$ case and found a perfect agreement.

\subsection{KS algebra for higher BGW tau-functions}\label{kso}
In the previous section, we construct the points of the Sato Grassmannian associated with the higher BGW tau-functions. In this section, we discuss the KS algebras for these points. Description of the higher BGW tau-functions by the KS operators was initiated by Mironov--Morozov--Semenoff \cite{Mironov:1994mv}, however, their description is incomplete. Here we provide a complete description, in particular, construct a canonical pair of KS operators.
 
Consider the operators
 \begin{equation}\label{ks}
{\mathtt a}_m={z}\frac{\partial}{\partial z}-\frac{m}{2}+\frac{z^m}{\hbar},\quad \quad{\mathtt b}_m=z^{m+1}
\end{equation}
satisfying the commutation relation
\begin{equation}
[{\mathtt a}_m,{\mathtt b}_m]=(m+1){\mathtt b}_m.
\end{equation}
Using the integration by parts in the integral representation of $ \Phi_j^{(m)}$, see \eqref{bv}, one can show that
\begin{equation}
\begin{split}\label{rec}
{\mathtt a}_m\cdot \Phi_j^{(m)} &={(j-1)(m+1)}\Phi_j^{(m)}+\frac{1}{\hbar}\Phi_{j+m}^{(m)},\\
{\mathtt b}_m\cdot \Phi_j^{(m)}&=j (m+1)\hbar \Phi_{j+1}^{(m)}+  \Phi_{j+m+1}^{(m)}.
\end{split}
\end{equation}
Therefore, the operators ${\mathtt a}_m$ and  ${\mathtt b}_m$ stabilize the point  \eqref{Gp} of the Sato Grassmannian 
\begin{equation}
{\mathtt a}_m\cdot \mathcal{W}_{m} \in \mathcal{W}_{m},\quad \quad
{\mathtt b}_m\cdot \mathcal{W}_{m} \in \mathcal{W}_{m},
\end{equation}
and ${\mathtt a}_m\in {\mathcal A}_{{\mathcal W}_m} $, ${\mathtt b}_m\in {\mathcal A}_{{\mathcal W}_m} $ are the KS operators.

However, it easy to see that the operators ${\mathtt a}_m$ and ${\mathtt b}_m$ do not generate the KS algebra ${\mathcal A}_{{\mathcal W}_m}$, therefore they do not
specify the point $\mathcal{W}_{m}$ in the Sato Grassmannian. To specify this point uniquely we need to consider some other KS operators.
\begin{remark}
Using integration by parts of the basis vectors \eqref{bv} we have
\begin{equation}\label{rec1}
\frac{\hbar}{{\mathtt b}_m}{\mathtt a}_m\cdot \Phi_j^{(m)}= \Phi_{j-1}^{(m)},
\end{equation}
which implies 
\begin{equation}
\frac{\hbar}{{\mathtt b}_m}{\mathtt a}_m \cdot \Phi_1^{(m)}= \Phi_{0}^{(m)}\notin\mathcal{W}_{m}.
\end{equation}
Hence, ${\mathtt b}_m^{-1}{\mathtt a}_m$  is not a KS operator for $\mathcal{W}_{m}$ \cite{Mironov:1994mv}.
\end{remark}
Let us consider the operators 
\begin{equation}\label{cex1}
\begin{split}
{\mathtt c}_m=\frac{\hbar}{{\mathtt b}_m}{\mathtt a}_m^2, \quad \quad
{\mathtt d}_m=\frac{1}{\hbar\, {\mathtt a}_m} {\mathtt b}_m.
\end{split}
\end{equation} 
These operators belong to ${\mathcal D}$,
\begin{equation}\label{cex}
\begin{split}
{\mathtt c}_m &=\hbar z^{1-m}\frac{\partial^2}{\partial z^2}+\frac{z^{m-1}}{\hbar}+\frac{m^2\hbar }{4z^{m+1}}+\left(2 +\hbar\frac{1-m}{z^m}\right)\frac{\partial}{\partial z},\\ 
{\mathtt d}_m&=\frac{1}{z^{m}\left(1+\hbar z^{-m}\left(z\frac{\p}{\p z}-\frac{m}{2}\right)\right)}z^{m+1}.
\end{split}
\end{equation} 
The operator ${\mathtt d}_m$ is given by a formal series
\be\label{dser}
{\mathtt d}_m=\sum_{k=0}^{\infty}\left(-\hbar z^{-m}\left(z\frac{\p}{\p z}-\frac{m}{2}\right)\right)^kz.
\ee

Combining (\ref{rec}) with (\ref{rec1}) we have
\begin{equation}\label{dKS}
\begin{split}
{\mathtt c}_m\cdot \Phi_j^{(m)}&=(j-1)(m+1)\Phi_{j-1}^{(m)}+\frac{1}{\hbar}\Phi_{j+m-1}^{(m)},\\
{\mathtt d}_m\cdot \Phi_j^{(m)}&= \Phi_{j+1}^{(m)}.
\end{split}
\end{equation} 
Therefore, these operators are KS operators for ${\mathcal W}_m$. The operator ${\mathtt c}_{\BGW}$ considered in Section \ref{S3.3} corresponds to the case $m=1$, ${\mathtt c}_{1}=\hbar{\mathtt c}_{\BGW}$.

The operators ${\mathtt c}_m$ and ${\mathtt d}_m$ satisfy the commutation relation
\be
[{\mathtt c}_m,{\mathtt d}_m]=m+1.
\ee

\begin{remark}
KS operators are related to the certain relations between the Lax operator $L$ and the Orlov-Schulman operator $M$, see e.g. \cite{Cafasso}. Existence of the KS operators
${\mathtt a}_m$ and ${\mathtt b}_m$ implies, in particular, that the operators
\be
P=L^{m+1},\quad \quad Q=ML+\frac{L^m}{\hbar}
\ee
satisfy
\be
(P)_-=(Q)_-=0
\ee
and a version of Douglas's string equation
\be
[P,Q]=(m+1)P.
\ee

Similarly for the operators ${\mathtt c}_m$ and ${\mathtt d}_m$ we consider
\be
\tilde{P}=\hbar M^2 L^{1-m}+\frac{L^{m-1}}{\hbar}+\frac{m^2\hbar }{4L^{m+1}}+M\left(2 +\hbar\frac{1-m}{L^m}\right).
\ee
and
\be
\tilde{Q}=L\sum_{k=0}^{\infty}\left(-\hbar \left(ML-\frac{m}{2}\right)L^{-m}\right)^k.
\ee
These operators also satisfy
\be
(\tilde{P})_-=(\tilde{Q})_-=0
\ee
and a version of Douglas's string equation
\be
[\tilde{P},\tilde{Q}]=(m+1).
\ee
\end{remark}

\begin{proposition}\label{P4.2}
For $m\geq 2$ the KS operators ${\mathtt c}_m$ and ${\mathtt d}_m$ completely specify the point ${\mathcal W}_m$ of the Sato Grassmannian. Namely, the quantum spectral curve operator 
is given by
\be\label{C1}
{\mathtt P}_{{\mathcal W}_m}=\frac{1}{m+1}\left( {\mathtt c}_m-\frac{1}{\hbar}({\mathtt d}_m)^{m-1}\right).
\ee
For $m=2$ we also have
\begin{equation}\label{C2}
\begin{split}
{\mathtt Q}_{{\mathcal W}_2}&=\frac{1}{3}\left(\hbar {\mathtt c}_2+2{\mathtt d}_2\right),\\
\end{split}
\end{equation}
for $m\geq3$ the operator ${\mathtt d}_m$ coincides with a canonical KS operator
\begin{equation}\label{dasQ}
\begin{split}
{\mathtt Q}_{{\mathcal W}_m}&={\mathtt d}_m.
\end{split}
\end{equation}
\end{proposition}
\begin{proof}
The operators $ {\mathtt c}_m$ and $ {\mathtt d}_m$ are KS operators for ${\mathcal W}_m$, hence their combinations in the right hand side of \eqref{C1},  \eqref{C2}, and \eqref{dasQ} 
are also KS operators. Therefore, it remains only to show that these operators are of the form $\frac{\p}{\p z}+{\mathcal D}_-$ (for \eqref{C1}) and $z+{\mathcal D}_-$ (for \eqref{C2} and \eqref{dasQ}),
by uniqueness it will immediately imply that they are canonical. 

From \eqref{cex} and \eqref{dser} for $m\geq 2$ we have
\begin{equation}
\begin{split}
 {\mathtt c}_m-\frac{1}{\hbar}({\mathtt d}_m)^{m-1}&=\frac{z^{m-1}}{\hbar}+2\frac{\p}{\p z} -\frac{1}{\hbar}\left(z-\frac{\hbar}{z^m}\left(\frac{\p}{\p z}-\frac{m}{2}\right)z\right)^{m-1}+\dots\\
&=(m+1)\frac{\p}{\p z}+\dots,
\end{split}
\end{equation}
where by $\dots$ we denote elements of ${\mathcal D}_-$. This implies \eqref{C1}.

From \eqref{dser} for $m\geq 3$ we have ${\mathtt d}_m\in z+{\mathcal D}_-$. 
From   \eqref{cex} and \eqref{dser} for $m=2$ we have
\begin{equation}
\begin{split}
\hbar {\mathtt c}_2&\in z+2\hbar \frac{\p}{\p z}+{\mathcal D}_-,\\
{\mathtt d}_2&\in z-\hbar\frac{\p}{\p z}+{\mathcal D}_-,
\end{split}
\end{equation}
therefore
\begin{equation}
\begin{split}
\frac{1}{3}\left(\hbar {\mathtt c}_2+2{\mathtt d}_2\right)\in  z+{\mathcal D}_-.
\end{split}
\end{equation}
This completes the proof.
\end{proof}

Together with Proposition \ref{Prop3.1} this theorem describes canonical KS operators for all $\tau^{(m)}$ with $m>0$.

\begin{remark}
From  \eqref{dKS} and \eqref{dasQ} we see that for $m\geq 3$ the basis given by Definition \ref{def1} is distinguished, see \cite[Section 2.4]{H3_2}
\be
\check{\Phi}_j^{(m)}=\Phi_j^{(m)}.
\ee
Then, Sato's group element associated with the tau-function $\tau^{(m)}$ for $m\geq 3$ is given by the asymptotic expansion of the integral near the critical point $\varphi_0=z^{-1}$
\be
{\mathtt G}_{{\mathcal W}_m}=\sqrt{\frac{2\pi {z}^m}{\hbar}}e^{-\frac{z^m}{m\hbar}}\frac{1}{2 \pi \iota}\int_{\gamma}e^{\frac{1}{m(m+1)\hbar}\left(m z^{m+1}\varphi+\frac{1}{\varphi^{m}}\right)+\varphi^{-1}\frac{\p}{\p z}}d \varphi.
\ee

For $m=2$ from \eqref{C2} we conclude that the distinguished basis is given by
\be
\check{\Phi}_j^{(2)}=\sqrt{\frac{2\pi}{\hbar}}ze^{-\frac{z^2}{2\hbar}}\frac{1}{2 \pi \iota}\int_{\gamma}e^{\frac{1}{6\hbar}\left(m z^{3}\varphi+\frac{1}{\varphi^{2}}\right)} H_j(\varphi^{-1}) d \varphi,
\ee
where $H_j(z)$ can be expressed in terms of the Hermite polynomial, $H_j(z)=(z+\hbar \frac{\p}{\p z})^j\cdot 1$.
\end{remark}

For $m=2$, corresponding to the reduction to the Boussinesq hierarchy,  we have
\begin{equation}
\begin{split}
{\mathtt Q}_{{\mathcal W}_2}&=\frac{1}{3}\left(\frac{\hbar^2}{z}\frac{\p^2}{\p z^2}+z+\frac{\hbar^2}{z^3}+\left(2\hbar-\frac{\hbar^2}{z^2}\right)\frac{\p}{\p z}+\frac{2}{z^{2}\left(1+\hbar z^{-2}\left(z\frac{\p}{\p z}-\frac{m}{2}\right)\right)}z^{3}\right),\\
{\mathtt P}_{{\mathcal W}_2}&=\frac{1}{3\hbar}\left(\frac{\hbar^2}{z}\frac{\p^2}{\p z^2}+z+\frac{\hbar^2}{z^3}+\left(2\hbar-\frac{\hbar^2}{z^2}\right)\frac{\p}{\p z}-\frac{1}{z^{2}\left(1+\hbar z^{-2}\left(z\frac{\p}{\p z}-\frac{m}{2}\right)\right)}z^{3}\right).
\end{split}
\end{equation}

 \subsection{Quantum spectral curve}\label{S4.4}        
 
 Let 
 \be
 \widehat{x}=\frac{z^{m+1}}{m+1},\quad \quad \widehat{y}=\frac{\hbar}{z^m}\frac{\p}{\p z}.
 \ee
These operators satisfy the canonical commutation relation
 \be
 \left[\widehat{y}, \widehat{x}\right]=\hbar.
 \ee
 
Using integration by parts it is easy to show that the functions ${\mathcal I}^{(m)}_{j}$, given by \eqref{bess}, satisfy
\be\label{qsc1}
\left((m+1)(m+1-j)\hbar\widehat{y}^m+(m+1)\widehat{x}\widehat{y}^{m+1}-1\right){\mathcal I}^{(m)}_{j-1}=0.
\ee
In particular, for $j=1$ we have
\be\label{qsc}
\left(m(m+1)\hbar\widehat{y}^m+(m+1)\widehat{x}\widehat{y}^{m+1}-1\right){\mathcal I}^{(m)}_{0}=0.
\ee   
This equation can be interpreted as a quantum spectral curve equation.
\begin{remark}
For $m=1$ this quantum spectral curve reduces to
\be
\left(2\hbar\widehat{y}+2\widehat{x}\widehat{y}^2-1\right){\mathcal I}_{0}^{(1)}=0,
\ee 
which was derived in \cite{KS2} for the BGW tau-function in a slightly different normalization.
\end{remark}  

The operators $\widehat{x}$ and $\widehat{y}$ are related to combinations of the operators ${\mathtt a}_m$ and ${\mathtt b}_m$ by a simple conjugation
\be
 {\mathtt  d}_m^{-1}=\sqrt{\frac{2\pi {z}^m}{\hbar}}e^{-\frac{z^m}{m\hbar}} \,\widehat{y} \, \sqrt{\frac{\hbar}{2\pi {z}^m}}e^{\frac{z^m}{m\hbar}},\quad  \frac{{\mathtt b}_m}{m+1}=\sqrt{\frac{2\pi {z}^m}{\hbar}}e^{-\frac{z^m}{m\hbar}}\, \widehat{x}\, \sqrt{\frac{\hbar}{2\pi {z}^m}}e^{\frac{z^m}{m\hbar}}.
\ee
Therefore, the basis vectors given by Definition \ref{def1} satisfy
\be\label{sqc1}
\left((m+1)(m+1-j)\hbar {\mathtt d}_m^{-m} +{\mathtt b}_m {\mathtt d}_m^{-m-1}-1\right)\cdot \Phi_j^{(m)}=0.
\ee
In general, the operator in the left hand side is not a KS one. However, these operators are closely related to the quantum spectral curve operators ${\mathtt P}_{{\mathcal W}_m}$, given by Proposition \ref{P4.2}, namely for $j=1$ we have
\be
\frac{1}{\hbar}{\mathtt d}_m^{m-1}\left((m+1)m\hbar {\mathtt d}_m^{-m} +{\mathtt b}_m {\mathtt d}_m^{-m-1}-1\right)={\mathtt c}_m-\frac{1}{\hbar}{\mathtt d}_m^{m-1}.
\ee

In the semi-classical limit the quantum spectral curve equation \eqref{qsc} reduces to the classical spectral curve
\be\label{csc}
(m+1)xy^{m+1}=1.
\ee
\begin{conjecture}\label{C11}
Generating functions of the higher BGW models $\tau^{(m)}$ are given by a version of the Chekhov--Eynard--Orantin topological recursion on the irregular spectral curve \eqref{csc},
described in \cite[Example 5.14]{Borot} with $r=m+1$ and $s=m$.
\end{conjecture}

 For $m=2$ the general solution to the equation \eqref{qsc} is
 \begin{multline}\label{soln}
 {\mathcal I}^{(2)}_{0}=\alpha_1\cdot {}_0F_2 \left(;\frac{1}{2},1;\frac{1}{216}\left(\frac{z^2}{\hbar}\right)^3\right)
 +\alpha_2z^3\cdot {}_0F_2 \left(;\frac{3}{2},\frac{3}{2};\frac{1}{216}\left(\frac{z^2}{\hbar}\right)^3\right)\\
 +\alpha_3 G^{2,0}_{0,3}\left(;\frac{1}{2},0,0;\frac{1}{216}\left(\frac{z^2}{\hbar}\right)^3\right),
 \end{multline}
where $\alpha_k$, $k=1,2,3$ are arbitrary constants. Here ${}_0F_2 $ is the generalized hypergeometric function and $G^{3,0}_{0,0}$ is the Meijer G-function. Comparing it with the small $z$ expansion, given by \eqref{smallzexp}, we conclude that $\alpha_1=1$ and $\alpha_2=\alpha_3=0$. Therefore
\be
\Phi_1^{(2)}=\sqrt{\frac{2\pi {z}^2}{\hbar}}e^{-\frac{z^2}{2\hbar}}\,\, {}_0F_2 \left(;\frac{1}{2},1;\frac{1}{216}\left(\frac{z^2}{\hbar}\right)^3\right),
\ee
where we take the asymptotic expansion of the generalized hypergeometric function at large positive argument. 

Similarly, for $j=2$ we have a general solution of \eqref{qsc1},
\begin{multline}\label{soln1}
 {\mathcal I}^{(2)}_{1}=\gamma_1z^3\cdot {}_0F_2 \left(;1,\frac{3}{2};\frac{1}{216}\left(\frac{z^2}{\hbar}\right)^3\right)
 +\gamma_2\cdot {}_0F_2 \left(;\frac{1}{2},\frac{1}{2};\frac{1}{216}\left(\frac{z^2}{\hbar}\right)^3\right)\\
 +\gamma_3 G^{2,0}_{0,3}\left(;\frac{1}{2},\frac{1}{2},0;\frac{1}{216}\left(\frac{z^2}{\hbar}\right)^3\right),
 \end{multline}
where $\gamma_k$, $k=1,2,3$ are arbitrary constants. Comparing it with the small $z$ expansion, given by \eqref{smallzexp}, we conclude that $\gamma_1=(3\hbar)^{-1}$ and $\gamma_2=\gamma_3=0$. For all other values of $j$, the general solution to the equation (\ref{qsc1}) is 
\begin{multline}\label{solnn}
 {\mathcal I}^{(2)}_{j-1}=\beta_1\cdot {}_0F_2 \left(;\frac{1}{2},\frac{3}{2}-\frac{j}{2};\frac{1}{216}\left(\frac{z^2}{\hbar}\right)^3\right)
 +\beta_2z^3\cdot {}_0F_2 \left(;\frac{3}{2},2-\frac{j}{2};\frac{1}{216}\left(\frac{z^2}{\hbar}\right)^3\right)\\
 +\beta_3 z^{3j-3}{}_0F_2 \left(;\frac{1}{2}+\frac{j}{2},\frac{j}{2};\frac{1}{216}\left(\frac{z^2}{\hbar}\right)^3\right).
 \end{multline}
Comparing it with the small $z$ expansion, given by \eqref{smallzexp} and large z expansion given in\cite{DLMF}, we conclude that  $\beta_1=\beta_2=0$ and $$\beta_3=\frac{\left(3\hbar\right)^{1-j}}{\Gamma(j)}~.$$ 
 Hence all $\Phi_j^{(2)}$ can be written as an asymptotic expansion of the hypergeometric function,
 \be\label{Aseq}
\Phi_j^{(2)}=\sqrt{\frac{2\pi {z}^2}{\hbar}}e^{-\frac{z^2}{2\hbar}}\frac{z^{3j-3}}{\left(3\hbar\right)^{j-1}\Gamma(j)}{}_0F_2 \left(;\frac{1}{2}+\frac{j}{2},\frac{j}{2};\frac{1}{216}\left(\frac{z^2}{\hbar}\right)^3\right).
\ee
 \subsection{W-constraints}\label{S3.6}    
In this section, we construct a family of linear $W^{(3)}$-constraints satisfied by the tau-functions $\tau^{(m)}({\bf t})$. For $m=2$ these constraints completely specify the formal series $\tau^{(m)}({\bf t})$, for higher $m$ the family of the constraints should be completed to $W^{(m+1)}$. The higher constraints can be constructed with the same methods and will be considered elsewhere.  

We consider the constraints corresponding to the particular linear combinations of the KS operators ${\mathtt b}_m^k$, ${\mathtt b}_m^k{\mathtt a}_m$ and ${\mathtt b}_m^k{\mathtt a}_m^2$.
The construction is completely analogous to the construction of the constraints for the generalized Kontsevich model, see e.g., \cite{H3_2}.

Let
\begin{equation}\label{vir}
\begin{split}
{\mathtt j}_k^{(m)}&=\frac{1}{m+1}{\mathtt b}_m^k=\frac{1}{m+1}{\mathtt j}_{(m+1)k},\\
{\mathtt l}_k^{(m)}&=-\frac{1}{m+1}{\mathtt b}_m^k{\mathtt a}_m+\frac{k+1}{2}{\mathtt b}_m^k=\frac{1}{m+1}\left({\mathtt l}_{(m+1)k}-\frac{1}{\hbar}{\mathtt j}_{(m+1)k+m}\right),\\
{\mathtt m}_k^{(m)}&=\frac{1}{m+1}{\mathtt b}_m^k{\mathtt a}_m^2+(k+1){\mathtt b}_m^k\left({\mathtt a}_m+\frac{1}{6}(m+1)(k+2)\right)\\
&=\frac{1}{m+1}\left({\mathtt m}_{(m+1)k}-\frac{2}{\hbar}{\mathtt l}_{(m+1)k+m}+\frac{1}{\hbar^2}{\mathtt j}_{(m+1)k+2m}+C_{m}{\mathtt j}_{(m+1)k}\right)
\end{split}
\end{equation}
with
\be
C_{m}=\frac{m(m+2)}{12}.
\ee
Recall that the operators ${\mathtt j}_k$, ${\mathtt l}_k$, and ${\mathtt m}_k$ are given by \eqref{virw}.

Then ${\mathtt j}_k^{(m)}$ for $k\geq 1$, ${\mathtt l}_k^{(m)}$ for $k\geq 0$, and (recall that ${\mathtt c}_m=\frac{\hbar}{{\mathtt b}_m}{\mathtt a}_m^2$ is a KS operator) ${\mathtt m}_k^{(m)}$ for  $k\geq -1$ are the KS operators for ${\mathcal W}_m$. 
 These operators satisfy the commutation relations
\begin{equation}
\begin{split}
\left[{\mathtt j}_k^{(m)},{\mathtt j}_j^{(m)}\right]&=0,\\
\left[{\mathtt j}_k^{(m)},{\mathtt l}_j^{(m)}\right]&=k{\mathtt j}_{k+j}^{(m)},\\
\left[{\mathtt l}_k^{(m)},{\mathtt l}_j^{(m)}\right]&=(k-j){\mathtt l}_{k+j}^{(m)},\\
\left[{\mathtt j}_k^{(m)},{\mathtt m}_j^{(m)}\right]&=2k{\mathtt l}_{k+j}^{(m)},\\
\left[{\mathtt l}_k^{(m)},{\mathtt m}_j^{(m)}\right]&=(2k-j){\mathtt m}_{k+j}^{(m)}+\frac{1}{6}k(m+1)^2(k^2-1){\mathtt j}_{k+j}^{(m)}.
\end{split}
\end{equation}
Let us use Lemma \ref{lemma33} to construct  operators annihilating the tau-function $\tau^{(m)}$. The only non-trivial part is to find the corresponding eigenvalues. They can be computed using the commutation relations between the operators. Let us consider the operators
\begin{equation}\label{Wopm}
 \begin{aligned}
\widehat{\mathcal J}_{k}^{(m)}&=\frac{1}{m+1}\widehat{J}_{(m+1)k}, && k\geq 1,\\
\widehat{\mathcal L}_k^{(m)}&=\frac{1}{m+1}\left(\widehat{L}_{(m+1)k} -\frac{1}{\hbar}\widehat{J}_{(m+1)k+m}+\frac{1}{2}C_{m}\delta_{k,0}\right),&& k\geq 0,\\
\widehat{\mathcal M}_k^{(m)}&=\frac{1}{m+1}\left(\widehat{M}_{(m+1)k}-\frac{2}{\hbar}  \widehat{L}_{(m+1)k+m}+\frac{1}{\hbar^2}\widehat{J}_{(m+1)k+2m}+C_{m} \widehat{J}_{(m+1)k}\right), \quad && k\geq -1.
\end{aligned}
\end{equation}
Recall that $\widehat{J}_0=0$.
These operators satisfy the following commutation relations
\begin{equation}
 \begin{aligned}
\left[{\widehat{\mathcal J}}_k^{(m)},~ { \widehat{\mathcal J}}_{k'}^{(m)} \right]&=0,  &&k,k'\geq 1,\\
\left[{\widehat{\mathcal J}}_k^{(m)},~ { \widehat{\mathcal L}}_{k'}^{(m)} \right]&=k{\widehat{\mathcal J}}_{k+k'}^{(m)}, &&k\geq 1,k'\geq 0,\\
\left[{\widehat{\mathcal L}}_k^{(m)},~ { \widehat{\mathcal L}}_{k'}^{(m)} \right]&=(k-k')\widehat{\mathcal L}_{k+k'}^{(m)}, &&k,k'\geq 0,\\
\left[{\widehat{\mathcal J}}_k^{(m)},~ {\widehat{\mathcal M}}_{k'} ^{(m)}\right]&=2k{\widehat{\mathcal L}}_{k+k'}^{(m)}, &&k\geq 1,k'\geq -1,\\
\left[{\widehat{\mathcal L}}_k^{(m)},~ { \widehat{\mathcal M}}_{k'}^{(m)}\right]&=(2k-k'){\widehat{\mathcal M}}_{k+k'}^{(m)}+\frac{1}{6}k(m+1)^2(k^2-1){\widehat{\mathcal J}}_{k+k'}^{(m)},\,\,\,\, &&k\geq 0,k'\geq -1.
\end{aligned}
\end{equation}
All operators \eqref{Wopm} can be obtained by a commutation of the other operators form the same family. Therefore, the eigenvalues for these operators are trivial, and we have the following theorem.
\begin{theorem}\label{T2}
The higher BGW tau-functions $\tau^{(m)}$ satisfy the $W^{(3)}$-constraints
\begin{equation}\label{Constrm}
\begin{split}
\widehat{\mathcal J}_{k}^{(m)}\cdot \tau^{(m)}&=0,\quad \quad k\geq1,\\
\widehat{\mathcal L}_{k}^{(m)}\cdot \tau^{(m)}&=0,\quad \quad k\geq0,\\
\widehat{\mathcal M}_{k}^{(m)}\cdot \tau^{(m)}&=0, \quad \quad k\geq -1.
\end{split}
\end{equation}
\end{theorem}
The first line of \eqref{Constrm} describes the $(m+1)$-reduction of the KP hierarchy (aka Gelfand--Dickey hierarchy).

\begin{remark}
For $m=1$ the constraints given by the first and the second lines of \eqref{Constrm} coincide with \eqref{rc} and \eqref{vc} respectively. 
\end{remark}
For instance, let us put $m=2$. In this case $\tau^{(2)}$ is a tau-function of the Boussinesq hierarchy, satisfying
\begin{equation}
\begin{split}\label{contm}
\widehat{\mathcal J}^{(2)}_{k}\cdot \tau^{(2)}&=0,\quad k\geq1,\\
\widehat{\mathcal L}^{(2)}_{k}\cdot \tau^{(2)}&=0,\quad k\geq0,\\
\widehat{\mathcal M}^{(2)}_{k}\cdot \tau^{(2)}&=0, \quad k\geq -1,
\end{split}
\end{equation}
where the operators are given by
\begin{equation}
\begin{split}
\widehat{\mathcal J}^{(2)}_{k}&=\frac{1}{3}\widehat{J}_{3k},\\
\widehat{\mathcal L}^{(2)}_k&=\frac{1}{3}\left(\widehat{L}_{3k} -\frac{1}{\hbar}\widehat{J}_{3k+2}+\frac{1}{3}\delta_{k,0}\right) ,\\
\widehat{\mathcal M}^{(2)}_k&=\frac{1}{3}\left(\widehat{M}_{3k}-\frac{2}{\hbar}  \widehat{L}_{3k+2}+\frac{1}{\hbar^2}\widehat{J}_{3k+4}+\frac{2}{3}\widehat{J}_{3k}\right).
\end{split}
\end{equation}
These operators satisfy the commutation relations
\begin{equation}
 \begin{aligned}
\left[{\widehat{\mathcal J}}_k^{(2)},~ { \widehat{\mathcal J}}_{k'}^{(2)} \right]&=0, &&k,k'\geq 1,\\
\left[{\widehat{\mathcal J}}_k^{(2)},~ { \widehat{\mathcal L}}_{k'}^{(2)} \right]&=k{\widehat{\mathcal J}}_{k+k'}^{(2)}, &&k\geq 1,k'\geq 0,\\
\left[{\widehat{\mathcal L}}_k^{(2)},~ { \widehat{\mathcal L}}_{k'}^{(2)} \right]&=(k-k')\widehat{\mathcal L}_{k+k'}^{(2)},&&k,k'\geq 0,\\
\left[{\widehat{\mathcal J}}_k^{(2)},~ {\widehat{\mathcal M}}_{k'}^{(2)} \right]&=2k{\widehat{\mathcal L}}_{k+k'}^{(2)},&&k\geq 1,k'\geq -1,\\
\left[{\widehat{\mathcal L}}_k^{(2)},~ { \widehat{\mathcal M}}_{k'}^{(2)}\right]&=(2k-k'){\widehat{\mathcal M}}_{k+k'}^{(2)}+\frac{3}{2}k(k^2
-1){\widehat{\mathcal J}}_{k+k'}^{(2)},\quad\quad \,&&k\geq 0,k'\geq -1.
\end{aligned}
\end{equation}
In the next section we use the constraints \eqref{contm} to derive the cut-and-join description of the tau-function $\tau^{(2)}$.

\subsection{Cut-and-join operators and algebraic topological recursion}\label{S4.6} In the previous section, we construct the $W^{(3)}$-constraints satisfied by the higher BGW tau-functions. For $m=2$ these constraints completely specify the formal series $\tau^{(2)}$ and can be solved. There are different ways to solve them, let us describe the solution in terms of the cut-and-join operators. This type of solution for the Virasoro constraints was introduced by Morozov and Shakirov for the Hermitian matrix model \cite{Morozov:2009xk}. The solution leads to the {\em algebraic topological recursion} \cite{H3_3}. We can rewrite the constraints \eqref{contm} as
\begin{equation}
\begin{split}
\frac{\p}{\p t_{3k+3}} \tau^{(2)}&=0,\\
\frac{\p}{\p t_{3k+2}}\tau^{(2)}&=\hbar\left(\widehat{L}_{3k}+\frac{1}{3}\delta_{k,0}\right)\cdot \tau^{(2)},\\
\frac{\p}{\p t_{3k+1}} \tau^{(2)}&=\left(2\hbar  \widehat{L}_{3k-1}- \hbar^2\left(\widehat{M}_{3k-3}+\frac{2}{3}\widehat{J}_{3k-3}\right)\right)\cdot \tau^{(2)},
\end{split}
\end{equation}
for $k \geq 0$. Combining these constraints we have
\begin{equation}\label{eulcut}
\widehat{D}\cdot \tau^{(2)} =\left(\hbar\widehat{W}_1 +\hbar^2\widehat{W}_2 \right)\cdot \tau^{(2)},
\end{equation}
where the cut-and-join operators $\widehat{W}_1$ and $\widehat{W}_2$ are given by
\begin{equation}
\begin{split}
\widehat{W}_1 &=\frac{2}{3}t_2+\sum_{k=0}^{\infty}\left((3k+2)t_{3k+2} \widehat{L}_{3k}+2(3k+1)t_{3k+1}  \widehat{L}_{3k-1} \right),\\
\widehat{W}_2 &=-2 t_3 t_1-\sum_{k=0}^{\infty}(3k+1)t_{3k+1} \widehat{ M}_{3k-3},
\end{split}
\end{equation}
and $\widehat{D}$ is the Euler operator \eqref{Eul}. 
\begin{remark}
These operators do not commute,
\be
\left[\widehat{W}_1,\widehat{W}_2\right]\neq 0.
\ee
\end{remark}
\begin{remark}
Operators $\widehat{W}_1$ and $\widehat{W}_2$ such that \eqref{eulcut} is satisfied are not unique. We do not know how to fix this ambiguity.
\end{remark}
From Lemma \ref{lemma31} we have
\begin{theorem}
\begin{equation}\label{TR2}
2 \hbar \frac{\p}{\p \hbar} \tau^{(2)} =\left(\hbar\widehat{W}_1 +\hbar^2\widehat{W}_2 \right)\cdot \tau^{(2)}.
\end{equation}
\end{theorem}

Let us consider the {\em topological expansion} of the higher BGW tau-function
\begin{equation}\label{texp}
\tau^{(m)}({\bf t })=1+\sum_{k=1}^{\infty}\hbar^k \tau^{(m,k)}({\bf t }),
\end{equation}
where  $\tau^{(m,k)}({\bf t })$ are polynomials in ${\bf t}$, which do not depend on $\hbar$.
 One can find these polynomials with the help of the
cut-and-join operators   which are differential operators acting on the time variables ${\bf t}$.  

For $m=2$ the topological expansion \eqref{texp} equation \eqref{TR2} leads to the algebraic topological recursion
\begin{equation}
2 k\, \tau^{(2,k)}= \widehat{W}_1 \cdot  \tau^{(2,k-1)}+\widehat{W}_2 \cdot \tau^{(2,k-2)}.
\end{equation}
This recursion allows us to find all the polynomials $ \tau^{(2,k)}$ recursively with the initial conditions  $\tau^{(2,-1)} =0$ and $\tau^{(2,0)}=1$.
For instance, the first coefficients of the topological expansion are given by
\begin{equation}
\tau^{(2)}({\bf t})=1+\frac{t_2}{3}  \hbar + \frac{14 t_2^2-3 t_1^4}{36} \hbar^2 + 
 \frac{182 t_2^3-117 t_1^4 t_2   - 432 t_1^2 t_4}{324} \hbar^3+O(\hbar^4).
\end{equation}
The solution can also be represented by the ordered exponential of the operators $\widehat{W}_1$ and $\widehat{W}_2$.
\begin{remark}
Similar cut-and-join description is known for the Kontsevich--Penner model \cite{Ocaj1,Ocaj2} and for the generalized Kontsevich model with monomial potential \cite{GKMcaj2,GKMcaj1}.
\end{remark}

For the higher $m$ one can combine the $W^{(m+1)}$-constraints as follows
\be
\widehat{D}\cdot \tau^{(m)}=\left(\hbar\widehat{W}_1 +\hbar^2\widehat{W}_2 +\ldots +\hbar^m\widehat{W}_m\right) \cdot\tau^{(m)}.
\ee
Here the form of the operators $\widehat{W}_k$ depend on the choice of $m$.
This equation, together with Lemma \ref{lemma31}, describes the algebraic topological recursion,  which allows to find all coefficients of the topological expansion \eqref{texp} recursively
\be
m k \, \tau^{(m,k)}= \widehat{W}_1 \cdot  \tau^{(m,k-1)}+\widehat{W}_2 \cdot \tau^{(m,k-2)}+\dots+\widehat{W}_m\cdot \tau^{(m,k-m)}.
\ee

In next section, we will consider a one-parametric deformation of the higher BGW models.

\section{Generalized higher BGW tau-functions}\label{S5}

In this section, we introduce a one-parametric deformation of the higher BGW tau-functions, considered in the previous section.  In particular, for this deformation we construct a complete set of the KS operators,  the quantum spectral curve, the $W^{(3)}$-constraints and for the simplest case with $m=2$ we construct the cut-and-join description. For the case of the original BGW tau-function ($m=1$) this deformation 
was investigated in \cite{KS2}, where it was called the generalized BGW tau-function.

\subsection{Matrix models and tau-functions}
Let us consider the matrix model \cite{Mironov:1994mv}
\begin{equation}\label{bgwp1}
Z^{(m,N)}=\frac{\displaystyle{\int [d \Phi] \exp\left(\frac{1}{\hbar}\text{Tr}\left(\frac{1}{m+1} \Lambda^{m+1} \Phi+\frac{1}{m(m+1)  \Phi^m}+\hbar(N- M)\log  \Phi   \right)     \right)}}{\displaystyle{\int [d \Phi] \exp\left(\frac{1}{\hbar}\text{Tr}~~\left(\frac{1}{m(m+1)  \Phi^m}+\hbar(N- M)\log \Phi   \right)     \right)}}.
 \end{equation}
 Here we use the notation from Section \ref{S4.1}. We consider $N$ as a formal parameter, in particular, we do not require that $N$ is an integer.
For $N=0$ this model coincides with the higher BGW model \eqref{Ebgwp}, 
\be
Z^{(m,0)}=Z^{(m)}.
\ee
Using the Harish-Chandra--Itzykson--Zuber formula for the integrals (\ref{bgwp1}) we have
\begin{equation}
Z^{(m,N)}=(-1)^{\frac{M(M-1)}{2}}\text{det}\left(\frac{\Lambda^{m+1}}{(m+1)\hbar}\right)^{N}\prod_{j=1}^M\Gamma(j-N)\frac{\text{det}_{i,j=1}^{M}\left(\mathcal I_{M-i}^{(m,N)}(\lambda_j)  \right)}{\Delta\left(\frac{\lambda^{m+1}}{(m+1)\hbar}\right)}~.
\end{equation}
 Here $\mathcal{ I}_{M-i}^{(m,N)}:=\mathcal I_{M-N-i}^{(m)}$ ,where the function $\mathcal I_{\nu}^{(m)}(z)$ is defined by \eqref{bess}.

To relate the matrix model in (\ref{bgwp1}) to the KP hierarchy we need to introduce a suitable prefactor. Namely,
we consider a function ${\mathcal C}_{m,N}^{-1}\, Z^{(m,N)}$, where
\be
{\mathcal C}_{m,N}=\frac{{ \Delta(\lambda)e^{\frac{\text{Tr} {\Lambda}^m}{m\hbar}}\text{det}\left(\frac{\Lambda^{m}}{(m+1)\hbar}\right)^{N}\prod_{i=1}^{M}\Gamma(j-N)}}{{ \Delta\left(\frac{\lambda^{m+1}}{(m+1)\hbar}\right) \left(\frac{2\pi}{\hbar}\right)^{\frac{M}{2}}\prod_{i=1}^{M}\lambda_i^\frac{m}{2}}}.
\ee
For this function a determinant formula holds,
\be
{\mathcal C}_{m,N}^{-1}\, Z^{(m,N)}=\frac{\text{det}_{i,j=1}^M\Phi_j^{(m,N)}(\lambda_i)}{\Delta(\lambda)}.
\ee

The functions $\Phi_j^{(m,N)}$ are given by

\begin{equation}
\begin{split}\label{bvN}
\Phi_j^{(m,N)}(z)&=z^N \Phi_{j-N}^{(m)}(z)\\
&=z^N\sqrt{\frac{2\pi {z}^m}{\hbar}}e^{-\frac{z^m}{m\hbar}}\mathcal I_{j-1}^{(m,N)}({z})\\
&=z^N\sqrt{\frac{2\pi {z}^m}{\hbar}}e^{-\frac{z^m}{m\hbar}}\frac{1}{2 \pi \iota}\int_{\gamma}e^{\frac{1}{m(m+1)\hbar}\left(m z^{m+1}\varphi+\frac{1}{\varphi^{m}}\right)}\frac{d \varphi}{\varphi^{j-N}}.
\end{split}
\end{equation}
By Remark \ref{R4.1} we have  $\Phi_j^{(m,N)}(z)=z^{j-1}\left(1+O(z^{-1})\right)$, therefore for a given $m\in {\mathbb Z}_{>0}$ and arbitrary parameter $N$ they constitute an admissible basis for a point of the Sato Grassmannian
\begin{equation}
\mathcal{W}_{m}^N=\sppan\left\{ \Phi_1^{(m,N)},\Phi_2^{(m,N)},\Phi_3^{(m,N)},\ldots\right\} \in \rm{Gr}^{(0)}_+,
\end{equation}
and
\begin{equation}
\tau^{(m,N)}\left(\left[\Lambda^{-1}\right]\right)=\frac{\text{det}_{i,j=1}^M\Phi_j^{(m,N)}(\lambda_i)}{\Delta(\lambda)}
\end{equation}
is a tau-function of the KP hierarchy. We call it the generalized higher BGW tau-function.
This is a one parameter deformation of the higher BGW tau-function,
\be
\tau^{(m,0)}=\tau^{(m)}.
\ee

The coefficients of the expansion of the basis vectors $\Phi_j^{(m,N)}(z)$ can be obtained from the coefficients of $\Phi_j^{(m)}(z)$, given in Appendix \ref{A1}, by a translation of $j$.
For example,  for $ j=1$ and $m=2$ we have
\begin{equation}
\Phi_1^{(2,N)}(z)=1+\hbar \frac{1 - 3 N - 3 N^2}{6 z^2}+\hbar^2 \frac{1 - 60 N - 15 N^2 + 30 N^3 + 9 N^4}{72 z^4}+\ldots.
\end{equation}

\begin{remark}
In \cite{KS2} for $m=1$ (the original generalized BGW model) it was proven that for $N\in {\mathbb Z}+1/2$ the tau-function is polynomial in ${\bf t}$ and is given by a certain Schur function. For the higher $m$ we do not see any similar polynomiality property.
\end{remark}

\subsection{KS algebra for $\tau^{(m,N)}$}
In Section \ref{kso}, we construct the KS algebra for the tau-function $\tau^{(m)}$ with $m>1$.  In this section, we construct a deformation of this KS algebra with the deformation parameter $N$. This deformation describes the tau-function    $\tau^{(m,N)}$.
Again, using the integral representation of the basis vectors $ \Phi_j^{(m,N)}$, one shows that the operators $\mathtt a_m$ and $\mathtt b_m$ given by \eqref{ks} satisfy the equations
\begin{eqnarray}
{\mathtt a}_m\cdot  \Phi_j^{(m,N)} &=&((j-1)(m+1)-Nm)\Phi_j^{(m,N)}+\frac{1}{\hbar}\Phi_{j+m}^{(m,N)},\nonumber\\
\mathtt b_m\cdot \Phi_j^{(m,N)}&=&(j-N) (m+1)\hbar \Phi_{j+1}^{(m,N)}+  \Phi_{j+m+1}^{(m,N)},
\end{eqnarray}
and are the KS operators for $\tau^{(m,N)}$
\begin{equation}
{\mathtt a}_m\cdot \mathcal{W}_{m}^N \in \mathcal{W}_{m}^N,\quad \quad 
{\mathtt b}_m\cdot \mathcal{W}_{m}^N \in \mathcal{W}_{m}^N.
\end{equation}
\begin{remark}
The operators ${\mathtt a}_m$ and ${\mathtt b}_m$ are KS for the family $\tau^{(m,N)}$ with all values of $N$, therefore they do not specify a tau-function and do not generate the KS algebra.
\end{remark}

From integration by parts we have
\begin{equation}
\frac{\hbar}{\mathtt b_m}\left({\mathtt a}_m -N\right)\cdot \Phi_j^{(m,N)}=\Phi_{j-1}^{(m,N)},
\end{equation}
but the operator $\frac{\hbar}{\mathtt b_m}\left({\mathtt a}_m -N\right)$ is not a KS operator for $\mathcal{W}_m^N$ because 
\begin{equation}
\frac{\hbar}{\mathtt b_m}\left({\mathtt a}_m -N\right) \cdot \Phi_1^{(m,N)}=\Phi_{0}^{(m,N)}\notin \mathcal{W}_m^N.
\end{equation}

Let us consider the operators 
\begin{equation}
\begin{split}
\mathtt c_m^{N}&=\frac{\hbar}{\mathtt b_m}({\mathtt a}_m-N)({\mathtt a}_m+Nm),\\
\mathtt{d}_m^N &=\frac{1}{\hbar\left({\mathtt a}_m -N\right)}{\mathtt b}_m.
\end{split}
\end{equation}
For $N=0$ they coincide with the operators, considered in Section \ref{kso}, $\mathtt c_m^{0}=\mathtt c_m$ and $\mathtt d_m^{0}=\mathtt d_m$.
These operators belong to ${\mathcal D}$,
\begin{equation}\label{Gcm}
\begin{split}
\mathtt c_m^{N}&=\hbar z^{1-m}\frac{\partial^2}{\partial z^2}+\frac{z^{m-1}}{\hbar}-\frac{m\hbar(2N+m)(2N-1) }{4z^{m+1}}+\frac{N(m-1)}{z}\\&+\left(2 +\hbar\frac{(1-N)(1-m)}{z^m}\right)\frac{\partial}{\partial z},\\
\mathtt{d}_m^N &=\frac{1}{z^{m}\left(1+\hbar z^{-m}\left(z\frac{\p}{\p z}-\frac{m}{2}-N\right)\right)}z^{m+1}.
\end{split}
\end{equation}
The operator $\mathtt{d}_m^N$ can be expressed as a formal series
\begin{equation}\label{Gdser}
{\mathtt d}_m^N=\sum_{k=0}^{\infty}\left(-\hbar z^{-m}\left(z\frac{\p}{\p z}-\frac{m}{2}-N\right)\right)^kz.
\end{equation}
These operators act on the basis vectors  $\Phi_j^{(m,N)}$ by a direct generalization of \eqref{dKS},
\begin{equation}\label{ks1}
\begin{split}
\mathtt c_m^{N}\cdot \Phi_j^{(m,N)}&=(j-1)(m+1)\Phi_{j-1}^{(m,N)}+\frac{1}{\hbar}\Phi_{j+m-1}^{(m,N)},\\
\mathtt{d}_m^N\cdot \Phi_j^{(m,N)}&=\Phi_{j+1}^{(m,N)}.
\end{split}
\end{equation}
Therefore, these operators are KS operators for $\mathcal{W}_{m}^N$.

The operators ${\mathtt c}_m^N$ and ${\mathtt d}_m^N$ satisfy the commutation relation
\be
[{\mathtt c}_m^N,{\mathtt d}_m^N]=m+1.
\ee
The KS operators ${\mathtt c}_m^N$ and ${\mathtt d}_m^N$ completely specify the point ${\mathcal W}_m^N$ of the Sato Grassmannian. 

\begin{proposition}\label{Prop5.0}
For the generalized BGW tau-function $\tau^{(2,N)}$ the canonical pair of KS operators is given by
 \begin{equation}
\begin{split}
{\mathtt P}_{{\mathcal W}_2^N}&=\sqrt{\hbar^{-1}{\mathtt c}_2^N} -\frac{1}{\hbar},\\
{\mathtt Q}_{{\mathcal W}_2^N}&=\left({\mathtt d}_2^N{\mathtt c}_2^N-N+\frac{1}{2}\right)\frac{1}{\sqrt{\hbar^{-1}{\mathtt c}_2^N}}.
\end{split}
\end{equation}
 \end{proposition}
 
\begin{proposition}\label{P5.1}
For $m\geq 2$ the KS operators ${\mathtt c}_m^N$ and ${\mathtt d}_m^N$ completely specify the point ${\mathcal W}_m^N$ of the Sato Grassmannian. Namely, the quantum spectral curve operator 
is given by
\be\label{GC1}
{\mathtt P}_{{\mathcal W}_m^N}=\frac{1}{m+1}\left( {\mathtt c}_m^N-\frac{1}{\hbar}({\mathtt d}_m^N)^{m-1}\right).
\ee
For $m=2$ we also have
\begin{equation}\label{GC2}
\begin{split}
{\mathtt Q}_{{\mathcal W}_2^N}&=\frac{1}{3}\left(\hbar {\mathtt c}_2^N+2{\mathtt d}_2^N\right),
\end{split}
\end{equation}
for $m\geq3$ the operator ${\mathtt d}_m^N$ coincides with a canonical KS operator
\begin{equation}\label{GdasQ}
\begin{split}
{\mathtt Q}_{{\mathcal W}_m^N}&={\mathtt d}_m^N.
\end{split}
\end{equation}
\end{proposition}
The proofs literally repeats the proofs of Propositions \ref{Prop3.1} and \ref{P4.2}.

\subsection{Quantum spectral curve}\label{S5.3}        
 
 Let 
 \be
 \widehat{x}=\frac{z^{m+1}}{m+1},\quad \quad \widehat{y}=\frac{\hbar}{z^{m}}\frac{\p}{\p z}.
 \ee
These operators satisfy the canonical commutation relation
 \be
 \left[\widehat{y}, \widehat{x}\right]=\hbar.
 \ee
 
Using integration by parts it is easy to show that the functions ${\mathcal I}^{(m,n)}_{j}={\mathcal I}^{(m)}_{j-N}$, given by \eqref{bess}, satisfy
\be\label{Nqsc1}
\left((m+1)(m+1-j+N)\hbar\widehat{y}^m+(m+1)\widehat{x}\widehat{y}^{m+1}-1\right){\mathcal I}^{(m,N)}_{j-1}=0.
\ee
In particular, for $j=1$ we have
\be\label{Nqsc}
\left((m+N)(m+1)\hbar\widehat{y}^m+(m+1)\widehat{x}\widehat{y}^{m+1}-1\right){\mathcal I}^{(m,N)}_{0}=0.
\ee   
This equation can be interpreted as a quantum spectral curve equation.

It is natural to consider another normalization of the parameter $N$ and introduce $S=N \hbar$. 
For instance, in this normalization for $m=1$ the quantum spectral curve \eqref{Nqsc} reduces to
\be
\left(2(S+\hbar)\widehat{y}+2\widehat{x}\widehat{y}^2-1\right){\mathcal I}_{0}^{(1,N)}=0.
\ee 
Then the classical spectral curve takes the form
\be
(m+1)(xy+S)y^m=1.
\ee

For $m=2$ the requiered solution to the equation (\ref{Nqsc1}) is given by (\ref{Aseq}) after the substitution of $j\rightarrow j-N$.
\subsection{W-constraints}
Let us construct the  $W^{(3)}$-constraints for the tau-function $\tau^{(m,N)}$. These constraints are given by the operators correspondent to the particular linear combinations of the  KS operators $\mathtt b^k_m$,  $\mathtt b_m^k{\mathtt a}_m$ and  $\mathtt b_m^k{\mathtt a}_m^2$.  The construction of the constraints is completely analogous to the construction of the constraints for the higher BGW tau-functions given in Section \ref{S3.6}. Namely, we leave the operators ${\mathtt j}_k^{(m)}$ and ${\mathtt l}_k^{(m)}$, which are are defined by equation (\ref{vir}), undeformed.  For the operators ${\mathtt m}_k^{(m)}$ we introduce the following deformation
\begin{equation}
\begin{split}
{\mathtt m}_k^{(m,N)}&=\frac{1}{m+1}\mathtt  b_m^k ({\mathtt a}_m-N)({\mathtt a}_m+Nm)+(k+1){\mathtt b}_m^k\left({\mathtt a}_m+\frac{1}{6}(m+1)(k+2)+\frac{A_{m,N}}{2}\right)\\
&=\frac{1}{m+1}\left({\mathtt m}_{(m+1)k}-\frac{2}{\hbar}{\mathtt l}_{(m+1)k+m}+\frac{1}{\hbar^2}{\mathtt j}_{(m+1)k+2m}+C_{m,N}{\mathtt j}_{(m+1)k}\right.\\ & \left. -A_{m,N}\left(\mathtt l_{(m+1)k}-\frac{1}{\hbar}{ \mathtt j_{(m+1)k+m}}\right)\right),
\end{split}
\end{equation}
where
\begin{equation}
C_{m,N}=C_{m}-N^2m, \quad \quad A_{m,N}=N(m-1).
\end{equation} 
Operators  ${\mathtt j}_k^{(m)}$ for $k \geq 1$, ${\mathtt l}_k^{(m)} $ for $k \geq 0$, and ${\mathtt m}_k^{(m,N)}$ for $k \geq -1$ are the KS operators for
$\mathcal{W}_m^N$. Let us apply Lemma \ref{lemma33} to construct  operators annihilating the tau-function $\tau^{(m,N)}$. 
Namely, consider the operators
\begin{equation}\label{W3N}
\begin{split}
\widehat{\mathcal J}_{k}^{(m,N)}&=\frac{1}{m+1}\widehat{J}_{(m+1)k},\\
\widehat{\mathcal L}_k^{(m,N)}&=\frac{1}{m+1}\left(\widehat{L}_{(m+1)k} -\frac{1}{\hbar}\widehat{J}_{(m+1)k+m}+\frac{1}{2}C_{m,N}\delta_{k,0}\right),\\
\widehat{\mathcal M}_k^{(m,N)}&=\frac{1}{m+1}\left(\widehat{M}_{(m+1)k}-\frac{2}{\hbar}  \widehat{L}_{(m+1)k+m}+\frac{1}{\hbar^2}\widehat{J}_{(m+1)k+2m}+C_{m,N}\widehat{J}_{(m+1)k}+\right.\\  &\left. -A_{m,N}\left(\widehat{L}_{(m+1)k}-\frac{1}{\hbar}{ \widehat{J}_{(m+1)k+m}}\right)-\frac{A_{m,N}}{3}\left(\frac	{1}{2}C_{m,N}+\frac{1}{12}(m^2+2m)\right)\delta_{k,0}\right).
\end{split}
\end{equation}
Commutation relations between these operators take the form
\begin{equation}
 \begin{aligned}
\left[{\widehat{\mathcal J}}_k^{(m,N)},~ { \widehat{\mathcal J}}_{k'}^{(m,N)} \right]&=0,~\quad&&k,k'\geq 1,\\
\left[{\widehat{\mathcal J}}_k^{(m,N)},~ { \widehat{\mathcal L}}_{k'}^{(m,N)} \right]&=k{\widehat{\mathcal J}}_{k+k'}^{(m,N)},&&k\geq 1,k'\geq 0,\\
\left[{\widehat{\mathcal L}}_k^{(m,N)},~ { \widehat{\mathcal L}}_{k'}^{(m,N)} \right]&=(k-k')\widehat{\mathcal L}_{k+k'}^{(m,N)},~\quad&&k,k'\geq 0,\\\left[{\widehat{\mathcal J}}_k^{(m,N)},~ {\widehat{\mathcal M}}_{k'}^{(m,N)} \right]&=2k{\widehat{\mathcal L}}_{k+k'}^{(m,N)}-kA_{m,N}{\widehat{\mathcal J}}_{k+k'}^{(m,N)},~\quad&&k\geq 1,k'\geq -1,\\
\left[{\widehat{\mathcal L}}_k^{(m,N)},~ { \widehat{\mathcal M}}_{k'}^{(m,N)}\right]&=(2k-k'){\widehat{\mathcal M}}_{k+k'}^{(2)}+k ~A_{m,N}{\widehat{\mathcal L}}_{k+k'}^{(m,N)}\\ &+\left(\frac{1}{6}k((m+1)^2(k^2-1)-12N^2m\right){\widehat{\mathcal J}}_{k+k'}^{(m,N)},
 \qquad~&&k\geq 0,k'\geq -1.
 \end{aligned}
\end{equation}
From these commutation relations we have the following theorem. 
\begin{theorem}\label{T3}
The generalized higher BGW tau-functions $\tau^{(m,N)}$ satisfy the $W^{(3)}$-constraints
\begin{equation}\label{contmn2}
\begin{split}
\widehat{\mathcal J}_{k}^{(m,N)}\cdot \tau^{(m,N)}&=0,\quad k\geq1,\\
\widehat{\mathcal L}_{k}^{(m,N)}\cdot \tau^{(m,N)}&=0,\quad k\geq0,\\
\widehat{\mathcal M}_{k}^{(m,N)}\cdot \tau^{(m,N)}&=0, \quad k\geq -1.
\end{split}
\end{equation}
\end{theorem}

In particular, we see that the $N$-deformations of the higher BGW models are tau-functions of the $m+1$-reduction of the KP hierarchy.

In the next section we will apply these constraints for $m=2$ to construction of the cut-and-join description for the tau-function $\tau^{(2,N)}$,
so let us consider this case in more detail. For $ m=2$,  $C_{2,N}=\frac{2}{3}-2N^2$, and the operators \eqref{W3N} took the form
\begin{equation}
\begin{split}
\widehat{\mathcal J}_{k}^{(2,N)}&=\frac{1}{3}\widehat{J}_{3k},\\
\widehat{\mathcal L}_k^{(2,N)}&=\frac{1}{3}\left(\widehat{L}_{3k} -\frac{1}{\hbar}\widehat{J}_{3k+2}+\frac{1}{2}C_{2,N}\delta_{k,0}\right) ,\\
\widehat{\mathcal M}_k^{(2,N)}&=\frac{1}{3}\left(\widehat{M}_{3k}-\frac{2}{\hbar}  \widehat{L}_{3k+2}+\frac{1}{\hbar^2}\widehat{J}_{3k+4} -N\left(\widehat{L}_{3k}-\frac{1}{\hbar}{ \widehat{J}_{3k+2}}\right)+C_{2,N}\widehat{J}_{3k}+\frac{N^3-N}{3}\delta_{k,0}\right).
\end{split}
\end{equation}

\subsection{Cut-and-join operators and algebraic topological recursion}
In this section, we describe the cut-and-join operators for the generalized higher BGW tau-function $\tau^{(2,N)}$. From Lemma \ref{lemma31} and equation \eqref{bvN} one has
\be\label{GH1}
\widehat{D}\cdot \tau^{(m,N)}=m  \hbar \frac{\p}{\p \hbar} \tau^{(m,N)}, 
\ee
where $\widehat{D}$ is the Euler operator.
Repeating the arguments of Section \ref{S4.6} we have
 \begin{equation}\label{GH2}
\widehat{D}\cdot \tau^{(2,N)} =\left(\hbar\widehat{W}_1 ^{(N)}+\hbar^2\widehat{W}_2^{(N)} \right)\cdot \tau^{(2,N)},
\end{equation}
where the cut-and-join operators $\widehat{W}_1^{(N)}$ and $\widehat{W}_2^{(N)}$ are
\begin{equation}
\begin{split}
\widehat{W}_1^{(N)} &=\sum_{k=0}^{\infty}\left((3k+2)t_{(3k+2)}\widehat{L}_{3k}+2(3k+1)t_{(3k+1)}\widehat{L}_{3k-1}\right)+2\left(\frac{1}{3}-N^2\right)t_2\\
 & -N\left(t_1^2+4t_4\frac{\partial}{\partial t_2}\right)\\
\widehat{W}_2 ^{(N)}&=-\sum_{k=0}^{\infty}(3k+1)t_{(3k+1)} \widehat{ M}_{3k-3}-\left(2-6N^2\right)t_3t_1 +N\left(4t_4 \widehat{L}_{0}+ t_1\widehat{L}_{-3}\right)\\
&-\frac{4(N^3-N)}{3}t_4.
\end{split}
\end{equation}
Combining \eqref{GH1} and \eqref{GH2} we arrive at the cut-and-join equation
\be\label{GH3}
 \frac{\p}{\p \hbar} \tau^{(2,N)}=\frac{1}{2} \left(\widehat{W}_1 ^{(N)}+\hbar\widehat{W}_2^{(N)} \right)\cdot \tau^{(2,N)}.
\ee

Let us consider the topological expansion of the generalized higher BGW tau-function
\begin{equation}
\tau^{(m,N)}({\bf t},\hbar)=1+\sum_{k=1}^{\infty}\hbar^k \tau^{(m,N,k)}(\bf t).
\end{equation}
Then the cut-and-join equation \eqref{GH3} is equivalent to the recursive relation
\begin{equation}
2 k \, \tau^{(2,N,k)}= \widehat{W}_1^N \cdot \tau^{(2,N,k-1)}+\widehat{W}_2^N \cdot \tau^{(2,N,k-2)}.
\end{equation}
This recursion allows us to find all the polynomials $ \tau^{(2,N,k)}$ for $k>0$ recursively with the initial conditions  $\tau^{(2,N,-1)} =0$ and $\tau^{(2,N,0)}=1$.  
The polynomials $\tau^{(2,k,N)}$ for $1\leq k \leq 4$ and the coefficients ${\mathcal F}^k$ of the expansion 
\be
\log \tau^{(2,N)} =\sum_{k=1}^\infty \hbar^k {\mathcal F}^k
\ee
for $1\leq k \leq 6$ are given in Appendix \ref{taufunction}.

\appendix
\section{Coefficients of the basis vectors}\label{A1}
\begin{flalign}
  \begin{aligned}
\Phi_{j,1}^{(2)}&=-\frac{1}{6} \left(3 j^2-9 j+5\right) 
\end{aligned}&&&
\end{flalign}

\begin{flalign}
  \begin{aligned}
\Phi_{j,2}^{(2)}&=\frac{1}{72} \left(9 j^4-66 j^3+129 j^2-36 j-35\right)
\end{aligned}&&&
\end{flalign}

\begin{flalign}
  \begin{aligned}
\Phi_{j,3}^{(2)}&=-\frac{1
}{1296} \left(27 j^6-351 j^5+1350 j^4-855 j^3-3312 j^2+3051 j+665\right)
\end{aligned}&&&
\end{flalign}

\begin{flalign}
  \begin{aligned}
\Phi_{j,4}^{(2)}&=\frac{1}{31104}\left(81 j^8-1620 j^7+10206 j^6-12852 j^5-77301 j^4+187740 j^3+82554 j^2\right.\\&\left.-285408 j+9625\right)
\end{aligned}&&&
\end{flalign}

\begin{flalign}
  \begin{aligned}
\Phi_{j,1}^{(m)}&=-\frac{1}{24} \left(12 j^2-12 j (m+1)+2 m^2+5 m+2\right)
\end{aligned}&&&
\end{flalign}

\begin{flalign}
  \begin{aligned}
\Phi_{j,2}^{(m)}&=\frac{1}{1152}\left(144 j^4-96 j^3 (5 m+1)+24 j^2 \left(20 m^2+5 m-4\right)\right.\\&\left.-24 j \left(6 m^3-m^2-9
m-2\right)+4 m^4-28 m^3-87 m^2-28 m+4\right) 
\end{aligned}&&&
\end{flalign}

\begin{flalign}
  \begin{aligned}
\Phi_{j,3}^{(m)}&=-\frac{1}{414720}\left(8640 j^6-8640 j^5 (7 m-1)+10800 j^4 \left(14 m^2-7 m-2\right)\right.\\&\left.-1440 j^3
\left(112 m^3-147 m^2-63 m+8\right)+180 j^2 \left(364 m^4-1288 m^3-567 m^2\right.\right.\\
&\left.\left.+392 m+76\right)-180 j \left(20 m^5-496 m^4-99 m^3+589 m^2+160 m-12\right)\right.\\&\left.-1112 m^6-6036 m^5+8934 m^4+38953 m^3+8934 m^2-6036 m-1112\right)
\end{aligned}&&&
\end{flalign}

\begin{flalign}
  \begin{aligned}
\Phi_{j,4}^{(m)}&=\frac{1}{39813120}\left( 103680 j^8-414720 j^7
(3 m-1)+725760 j^6 \left(8 m^2-7 m\right)\right.\\&\left.-48384 j^5 \left(274 m^3-489 m^2+39 m+26\right)\right.\\&\left.+30240 j^4 \left(500 m^4-1740 m^3+495 m^2+340 m-12\right)\right.\\&\left.-2880 j^3 \left(2588 m^5-19780 m^4+14453 m^3+9697 m^2-1892 m-404\right)\right.\\&
\left. +48 j^2\left(11488 m^6-547836 m^5+1007484 m^4+593353 m^3-411276 m^2\right.\right.\\&\left.\left.-114396
m+4288\right)+48 j\left(7784 m^7+55964 m^6-443022 m^5-154855 m^4\right.\right.\\&\left. \left.+500081 m^3+134826 m^2-34468 m
-5560
\right)
 -9136 m^8+430496 m^7\right.\\&\left.+2055608 m^6
 -1245112 m^5-8204587 m^4-1245112 m^3+2055608 m^2+430496 m\right.\\&\left.-9136
\right)
\end{aligned}&&&
\end{flalign}

\section{Topological expansion of the tau-function $\tau^{(2)}$}\label{taufunction1}
\begin{flalign}
  \begin{aligned}
 \tau^{(2,1)}=\frac{1}{3} t_2
\end{aligned}&&&
\end{flalign}
\begin{flalign}
  \begin{aligned}
 \tau^{(2,2)}= -\frac{1}{12}  t_1^4+\frac{7}{18}  t_2^2
\end{aligned}&&&
\end{flalign}
\begin{flalign}
  \begin{aligned}
 \tau^{(2,3)}=-\frac{13}{36} t_1^4 t_2+\frac{91}{162} t_2^3-\frac{4}{3} t_1^2 t_4
\end{aligned}&&&
\end{flalign}
\begin{flalign}
  \begin{aligned}
  \tau^{(2,4)}=\frac{1}{288}{t_1^8}-\frac{247}{216} t_1^4 t_2^2+\frac{1729}{1944} t_2^4-\frac{76}{9} t_1^2 t_2 t_4-\frac{26}{9} t_4^2-\frac{25}{9} t_1^3
t_5-\frac{49 }{9}t_1 t_7
\end{aligned}&&&
\end{flalign}
\begin{flalign}
  \begin{aligned}
\tau^{(2,5)}=&\frac{25}{864} t_1^8 t_2-\frac{6175}{1944} t_1^4 t_2^3+\frac{8645}{5832} t_2^5+\frac{5}{9} t_1^6 t_4-\frac{950}{27} t_1^2 t_2^2 t_4-\frac{650}{27}
t_2 t_4^2-\frac{625}{27} t_1^3 t_2 t_5\\ &-\frac{400}{9} t_1 t_4 t_5-\frac{1225}{27} t_1 t_2 t_7-\frac{200}{9} t_1^2 t_8-\frac{1225 }{81}t_{10}
\end{aligned}&&&
\end{flalign}
\begin{flalign}
  \begin{aligned}
\tau^{(2,6)}=&-\frac{1}{10368}t_1^{12}+\frac{775}{5184} t_1^8 t_2^2-\frac{191425}{23328} t_1^4 t_2^4+\frac{267995}{104976} t_2^6+\frac{155}{27} t_1^6
t_2 t_4-\frac{29450}{243} t_1^2 t_2^3 t_4\\ &+\frac{1525}{54} t_1^4 t_4^2-\frac{10075}{81} t_2^2 t_4^2+\frac{85}{108} t_1^7 t_5-\frac{19375}{162} t_1^3
t_2^2 t_5-\frac{12400}{27} t_1 t_2 t_4 t_5-\frac{225}{2} t_1^2 t_5^2\\ &+\frac{1225}{108} t_1^5 t_7-\frac{37975}{162} t_1 t_2^2 t_7-\frac{12250}{81} t_5 t_7-\frac{6200}{27}
t_1^2 t_2 t_8-\frac{12400}{81} t_4 t_8\\ &-\frac{37975}{243} t_2 t_{10}-\frac{33275}{243} t_1 t_{11}
\end{aligned}&&&
\end{flalign}
\begin{flalign}
  \begin{aligned}
\tau^{(2,7)}=&-\frac{37}{31104} t_1^{12} t_2+\frac{28675}{46656} t_1^8 t_2^3-\frac{1416545}{69984} t_1^4 t_2^5+\frac{1416545}{314928} t_2^7-\frac{1}{24}
t_1^{10} t_4+\frac{5735}{162} t_1^6 t_2^2 t_4\\ &-\frac{544825}{1458} t_1^2 t_2^4 t_4+\frac{56425}{162} t_1^4 t_2 t_4^2-\frac{372775}{729} t_2^3 t_4^2+\frac{12200}{27}
t_1^2 t_4^3+\frac{3145}{324} t_1^7 t_2 t_5\\ &-\frac{716875}{1458} t_1^3 t_2^3 t_5+\frac{3680}{27} t_1^5 t_4 t_5-\frac{229400}{81} t_1 t_2^2 t_4 t_5-\frac{2775}{2}
t_1^2 t_2 t_5^2-\frac{24700}{27} t_4 t_5^2\\ &+\frac{45325}{324} t_1^5 t_2 t_7-\frac{1405075}{1458} t_1 t_2^3 t_7+\frac{24500}{27} t_1^3 t_4 t_7-\frac{453250}{243}
t_2 t_5 t_7+\frac{6050}{81} t_1 t_6 t_7\\ &+\frac{610}{27} t_1^6 t_8-\frac{114700}{81} t_1^2 t_2^2 t_8-\frac{458800}{243} t_2 t_4 t_8-\frac{2953400}{1701} t_1
t_5 t_8+\frac{12100}{189} t_1 t_4 t_9\\ &+\frac{221725}{972} t_1^4 t_{10}-\frac{1405075}{1458} t_2^2 t_{10}+\frac{30250}{567} t_1 t_3 t_{10}-\frac{7619975}{5103}
t_1 t_2 t_{11}-\frac{111925}{243} t_{14}
\end{aligned}&&&
\end{flalign}
\begin{flalign}
  \begin{aligned}
\tau^{(2,8)}=&\frac{1}{497664}t_1^{16}-\frac{1591}{186624} t_1^{12} t_2^2+\frac{1233025}{559872} t_1^8 t_2^4-\frac{60911435}{1259712} t_1^4 t_2^6+\frac{60911435}{7558272}
t_2^8-\frac{43}{72} t_1^{10} t_2 t_4\\ &+\frac{246605}{1458} t_1^6 t_2^3 t_4-\frac{4685495}{4374} t_1^2 t_2^5 t_4-\frac{11005}{1296} t_1^8 t_4^2+\frac{2426275}{972}
t_1^4 t_2^2 t_4^2-\frac{16029325}{8748} t_2^4 t_4^2\\ &+\frac{524600}{81} t_1^2 t_2 t_4^3+\frac{29350}{27} t_4^4-\frac{145}{2592} t_1^{11} t_5+\frac{135235}{1944}
t_1^7 t_2^2 t_5-\frac{30825625}{17496} t_1^3 t_2^4 t_5\\ &+\frac{158240}{81} t_1^5 t_2 t_4 t_5-\frac{9864200}{729} t_1 t_2^3 t_4 t_5+\frac{516650}{81}
t_1^3 t_4^2 t_5+\frac{102775}{648} t_1^6 t_5^2-\frac{39775}{4} t_1^2 t_2^2 t_5^2\\ &-\frac{1062100}{81} t_2 t_4 t_5^2-\frac{14054125}{3402} t_1 t_5^3-\frac{922625}{567}
t_1 t_4 t_5 t_6-\frac{4865}{2592} t_1^9 t_7+\frac{1948975}{1944} t_1^5 t_2^2 t_7\\ &-\frac{60418225}{17496} t_1 t_2^4 t_7+\frac{1053500}{81} t_1^3 t_2
t_4 t_7+\frac{91750}{9} t_1 t_4^2 t_7+\frac{1556975}{486} t_1^4 t_5 t_7\\ &-\frac{9744875}{729} t_2^2 t_5 t_7-\frac{438625}{324} t_1 t_3 t_5 t_7+\frac{133100}{243}
t_1 t_2 t_6 t_7+\frac{1658125}{324} t_1^2 t_7^2+\frac{26230}{81} t_1^6 t_2 t_8\\ &-\frac{4932100}{729} t_1^2 t_2^3 t_8+\frac{777100}{243} t_1^4 t_4 t_8-\frac{9864200}{729}
t_2^2 t_4 t_8-\frac{102850}{81} t_1 t_3 t_4 t_8\\ &-\frac{121959575}{5103} t_1 t_2 t_5 t_8-\frac{414425}{567} t_1^2 t_6 t_8-\frac{6216400}{1701} t_8^2+\frac{238975}{567}
t_1 t_2 t_4 t_9-\frac{75625}{108} t_1^2 t_5 t_9\\ &+\frac{3025}{54} t_7 t_9+\frac{9534175}{2916} t_1^4 t_2 t_{10}-\frac{60418225}{13122} t_2^3 t_{10}+\frac{1467125}{3402}
t_1 t_2 t_3 t_{10}\\ &+\frac{18194800}{1701} t_1^2 t_4 t_{10}+\frac{15125}{162} t_6 t_{10}+\frac{645535}{1458} t_1^5 t_{11}-\frac{83620075}{8748} t_1
t_2^2 t_{11}-\frac{1031525}{2268} t_1^2 t_3 t_{11}\\ &-\frac{29514925}{5103} t_5 t_{11}+\frac{12984725}{3888} t_1^3 t_{13}+\frac{196625}{2268} t_3 t_{13}-\frac{16900675}{2916}
t_2 t_{14}
\end{aligned}&&&
\end{flalign}

\begin{flalign}
  \begin{aligned}
\tau^{(2,9)}=&\frac{49 }{1492992}t_1^{16} t_2-\frac{77959}{1679616} t_1^{12} t_2^3+\frac{12083645}{1679616} t_1^8 t_2^5-\frac{426380045}{3779136} t_1^4 t_2^7+\frac{2984660315}{204073344}
t_2^9\\ &+\frac{13}{7776} t_1^{14} t_4-\frac{2107}{432} t_1^{10} t_2^2 t_4+\frac{12083645}{17496} t_1^6 t_2^4 t_4-\frac{229589255}{78732} t_1^2 t_2^6
t_4-\frac{539245}{3888} t_1^8 t_2 t_4^2\\ &+\frac{118887475}{8748} t_1^4 t_2^3 t_4^2-\frac{157087385}{26244} t_2^5 t_4^2-\frac{23030}{27} t_1^6
t_4^3+\frac{12852700}{243} t_1^2 t_2^2 t_4^3+\frac{1438150}{81} t_2 t_4^4\\ &-\frac{7105}{7776} t_1^{11} t_2 t_5+\frac{6626515}{17496} t_1^7 t_2^3 t_5-\frac{302091125}{52488}
t_1^3 t_2^5 t_5-\frac{1645}{54} t_1^9 t_4 t_5+\frac{3876880}{243} t_1^5 t_2^2 t_4 t_5\\ &-\frac{120836450}{2187} t_1 t_2^4 t_4 t_5+\frac{25315850}{243}
t_1^3 t_2 t_4^2 t_5+\frac{858179800}{15309} t_1 t_4^3 t_5+\frac{5035975}{1944} t_1^6 t_2 t_5^2\\ &-\frac{1948975}{36} t_1^2 t_2^3 t_5^2+\frac{2078825}{81}
t_1^4 t_4 t_5^2-\frac{26021450}{243} t_2^2 t_4 t_5^2-\frac{13536875}{1701} t_1 t_3 t_4 t_5^2\\ &-\frac{289086625}{4374} t_1 t_2 t_5^3-\frac{11966900}{1701}
t_1 t_3 t_4^2 t_6-\frac{117687625}{5103} t_1 t_2 t_4 t_5 t_6-\frac{3554375}{972} t_1^2 t_5^2 t_6\\ &-\frac{865150}{243} t_1^2 t_4 t_6^2-\frac{238385}{7776}
t_1^9 t_2 t_7+\frac{95499775}{17496} t_1^5 t_2^3 t_7-\frac{592098605}{52488} t_1 t_2^5 t_7-\frac{6545}{9} t_1^7 t_4 t_7\\ &+\frac{25810750}{243}
t_1^3 t_2^2 t_4 t_7-\frac{1491325}{243} t_1 t_3^2 t_4 t_7+\frac{118457050}{729} t_1 t_2 t_4^2 t_7+\frac{76291775}{1458} t_1^4 t_2 t_5 t_7\\ &-\frac{477498875}{6561}
t_2^3 t_5 t_7-\frac{58911875}{2916} t_1 t_2 t_3 t_5 t_7+\frac{4603675}{27} t_1^2 t_4 t_5 t_7-\frac{223850}{243} t_1^5 t_6 t_7\\ &+\frac{801625}{486}
t_1 t_2^2 t_6 t_7-\frac{1355200}{243} t_1^2 t_3 t_6 t_7-\frac{2344375}{972} t_5 t_6 t_7+\frac{241372775}{2916} t_1^2 t_2 t_7^2\\ &+\frac{31950625}{729}
t_4 t_7^2-\frac{1001}{324} t_1^{10} t_8+\frac{642635}{243} t_1^6 t_2^2 t_8-\frac{60418225}{2187} t_1^2 t_2^4 t_8+\frac{38077900}{729} t_1^4 t_2 t_4
t_8\\ &-\frac{483345800}{6561} t_2^3 t_4 t_8-\frac{14235650}{729} t_1 t_2 t_3 t_4 t_8+\frac{16050400}{189} t_1^2 t_4^2 t_8+\frac{47878520}{5103} t_1^5 t_5 t_8\\ &-\frac{2884168900}{15309}
t_1 t_2^2 t_5 t_8-\frac{7184375}{1134} t_1^2 t_3 t_5 t_8-\frac{549410975}{10206} t_5^2 t_8-\frac{7498975}{729} t_1^2 t_2 t_6 t_8\\ &-\frac{18101600}{1701}
t_4 t_6 t_8+\frac{160884275}{2916} t_1^3 t_7 t_8-\frac{260150}{27} t_3 t_7 t_8-\frac{863220700}{15309} t_2 t_8^2\\ &-\frac{447700}{567} t_1^5
t_4 t_9+\frac{568700}{567} t_1 t_2^2 t_4 t_9-\frac{2665025}{567} t_1^2 t_3 t_4 t_9-\frac{30749125}{3402} t_1^2 t_2 t_5 t_9\\ &-\frac{196625}{126} t_4
t_5 t_9-\frac{1491325}{1134} t_1^3 t_6 t_9-\frac{378125}{972} t_2 t_7 t_9-\frac{20933000}{5103} t_1 t_8 t_9-\frac{2105425}{23328} t_1^8 t_{10}\\ &+\frac{467174575}{17496}
t_1^4 t_2^2 t_{10}-\frac{2960493025}{157464} t_2^4 t_{10}-\frac{1119250}{1701} t_1^5 t_3 t_{10}+\frac{892375}{567} t_1 t_2^2 t_3 t_{10}\\ &-\frac{3025000}{1701}
t_1^2 t_3^2 t_{10}+\frac{884406200}{5103} t_1^2 t_2 t_4 t_{10}+\frac{247596950}{5103} t_4^2 t_{10}+\frac{1560537625}{27216} t_1^3 t_5 t_{10}\\ &-\frac{13990625}{6804}
t_3 t_5 t_{10}+\frac{75625}{1458} t_2 t_6 t_{10}+\frac{168778325}{2187} t_1 t_7 t_{10}+\frac{212114815}{30618} t_1^5 t_2 t_{11}\\ &-\frac{25865822125}{551124}
t_1 t_2^3 t_{11}-\frac{1331000}{243} t_1^2 t_2 t_3 t_{11}+\frac{14998174675}{367416} t_1^3 t_4 t_{11}\\ &-\frac{198618475}{30618} t_3 t_4 t_{11}-\frac{30752455525}{367416}
t_2 t_5 t_{11}-\frac{266765675}{61236} t_1 t_6 t_{11}+\frac{869976575}{17496} t_1^3 t_2 t_{13}\\ &+\frac{5151575}{20412} t_2 t_3 t_{13}+\frac{616094375}{8748}
t_1 t_4 t_{13}+\frac{118134325}{11664} t_1^4 t_{14}-\frac{123005575}{2916} t_2^2 t_{14}\\ &-\frac{30779375}{8748} t_1 t_3 t_{14}+\frac{128714200}{5103}
t_1^2 t_{16}
\end{aligned}&&&
\end{flalign}

\section{Topological expansion of the tau-function $\tau^{(2,N)}$}\label{taufunction}
\begin{flalign}
  \begin{aligned}
\tau^{(2,N,1)}&=-\frac{1}{2} N t_1^2-\frac{1}{3} \left(3 N^2-1\right) t_2,
\end{aligned}&&&
\end{flalign}
\begin{flalign}
  \begin{aligned}
\tau^{(2,N,2)} &=\frac{1}{24} \left(3 N^2-2\right) t_1^4+\frac{1}{6} N \left(3 N^2-7\right) t_1^2 t_2+\frac{1}{18} \left(9N^4-24 N^2+7\right)
t_2^2\\
&+\frac{1}{3} N \left(2 N^2-3\right) t_4,
\end{aligned}&&&
\end{flalign}
\begin{flalign}
  \begin{aligned}
\tau^{(2,N,3)} &=-\frac{1}{48} N \left(N^2-2\right) t_1^6-\frac{1}{72} \left(9 N^4-45 N^2+26\right) t_1^4 t_2\\&
-\frac{1}{36} N \left(9
N^4-60 N^2+91\right) t_1^2 t_2^2-\frac{1}{162} \left(27 N^6-189 N^4+333 N^2-91\right) t_2^3\\&
-\frac{1}{6} \left(2 N^4-15 N^2+8\right) t_1^2 t_4-\frac{1}{9} N \left(6 N^4-35 N^2+39\right)
t_2 t_4\\&+\frac{5}{3} N \left(N^2-2\right) t_1 t_5,
\end{aligned}&&&
\end{flalign}
\begin{flalign}
  \begin{aligned}
\tau^{(2,N,4)}&=\frac{1}{1152}\left(3 N^4-12 N^2+4\right) t_1^8+\frac{1}{324} N \left(27 N^6-351 N^4+1413 N^2-1729\right) t_1^2 t_2^3\\&
+\frac{1}{144} N \left(3 N^4-25 N^2+38\right) t_1^6 t_2
+\frac{1}{432} \left(27 N^6-306 N^4+933 N^2-494\right) t_1^4 t_2^2
\\&+\frac{1}{1944}\left(81 N^8-1080 N^6+4590 N^4-6600 N^2+1729\right) t_2^4\\&+\frac{1}{18} \left(6 N^6-83 N^4+309 N^2-152\right) t_1^2 t_2 t_4+\frac{1}{72} N \left(6 N^4-85 N^2+174\right) t_1^4 t_4\\&+\frac{1}{18} \left(4 N^6-48 N^4+141 N^2-52\right) t_4^2
-\frac{7}{9} \left(3 N^4-15 N^2+7\right) t_1 t_7
\\&
+\frac{1}{54} N \left(18 N^6-219 N^4+782 N^2-741\right) t_2^2 t_4-\frac{5}{9} N \left(3 N^4-25 N^2+38\right)
t_1 t_2 t_5\\&
-\frac{5}{18} \left(3 N^4-18 N^2+10\right) t_1^3 t_5-\frac{1}{9} N \left(6 N^4-50 N^2+69\right) t_8
\end{aligned}&&&
\end{flalign}
\begin{flalign}
  \begin{aligned}
\tau^{(2,N,5)} &=-\frac{1}{11520}N \left(3 N^4-20 N^2+20\right) t_1^{10}\\&-\frac{1}{3456}\left(9 N^6-111 N^4+312 N^2-100\right) t_1^8
t_2\\&-\frac{1}{864} N \left(9 N^6-150 N^4+739 N^2-950\right) t_1^6 t_2^2\\&-\frac{1}{3888}\left(81 N^8-1593 N^6+10449 N^4-24807 N^2+12350\right) t_1^4 t_2^3\\&-\frac{1}{3888}N \left(81 N^8-1728 N^6+13014 N^4-40512 N^2+43225\right) t_1^2 t_2^4\\&-\frac{1}{29160}\left(243 N^{10}-5265 N^8+40770 N^6-134550 N^4+170187 N^2\right.\\&\left.-43225\right) t_2^5-\frac{1}{144} \left(2 N^6-43 N^4+174 N^2-80\right) t_1^6 t_4\\&-\frac{1}{216} N \left(18 N^6-405 N^4+2647 N^2-4350\right) t_1^4 t_2 t_4\\&-\frac{1}{108} \left(18 N^8-399 N^6+3002 N^4-8181 N^2+3800\right) t_1^2 t_2^2 t_4\\&-\frac{1}{486} N \left(54 N^8-1107 N^6+7821 N^4-21773 N^2+18525\right) t_2^3 t_4\\&-\frac{1}{36} N \left(4 N^6-96 N^4+821 N^2-1444\right) t_1^2
t_4^2\\&+\frac{1}{72} N \left(15 N^4-160 N^2+288\right) t_1^5 t_5\\&-\frac{1}{54} \left(12 N^8-244 N^6+1623 N^4-3681 N^2+1300\right) t_2 t_4^2\\&+\frac{5}{54} \left(9 N^6-129 N^4+480 N^2-250\right) t_1^3 t_2 t_5\\&+\frac{5}{54} N \left(9 N^6-150 N^4+739 N^2-950\right) t_1 t_2^2 t_5\\&+\frac{5}{9} \left(2 N^6-43 N^4+174 N^2-80\right) t_1 t_4 t_5-\frac{5}{18} N \left(9 N^4-75 N^2+107\right) t_5^2\\&+\frac{7}{18} N \left(3 N^4-35
N^2+67\right) t_1^3 t_7+\frac{7}{27} \left(9 N^6-120 N^4+396 N^2-175\right) t_1 t_2 t_7\\&+\frac{1}{18} \left(6 N^6-170 N^4+789 N^2-400\right) t_1^2 t_8\\&+\frac{1}{27} N \left(18 N^6-300 N^4+1457 N^2-1725\right) t_2 t_8\\&+\frac{7}{81} \left(9 N^6-135 N^4+441 N^2-175\right) t_{10}.
\end{aligned}&&&
\end{flalign}

\begin{flalign}
  \begin{aligned}
\mathcal{F}^1&=-\frac{1}{2} N t_1^2-\frac{1}{3} \left(3 N^2-1\right) t_2,
\end{aligned}&&&
\end{flalign}
\begin{flalign}
  \begin{aligned}
\mathcal{F}^2 &=-\frac{1}{12}t_1^4-N t_1^2 t_2-\frac{1}{3} \left(3 N^2-1\right) t_2^2+\frac{1}{3} N \left(2 N^2-3\right) t_4,
\end{aligned}&&&
\end{flalign}

\begin{flalign}
  \begin{aligned}
\mathcal{F}^3&=-\frac{1}{3} t_1^4 t_2-2 N t_1^2 t_2^2-\frac{4}{9} \left(3 N^2-1\right) t_2^3+\frac{2}{3} \left(3 N^2-2\right) t_1^2 t_4+\frac{4}{3}
N \left(2 N^2-3\right) t_2 t_4\\&+\frac{5}{3} N \left(N^2-2\right) t_1 t_5
\end{aligned}&&&
\end{flalign}
\begin{flalign}
  \begin{aligned}
\mathcal{F}^4 &=-t_1^4 t_2^2-4 N t_1^2 t_2^3-\frac{2}{3} \left(3 N^2-1\right) t_2^4+\frac{5}{3} N t_1^4 t_4+4 \left(3 N^2-2\right) t_1^2
t_2 t_4\\&+4N(2N^2-3)t_2^2t_4-\frac{2}{9} \left(9 N^4-33 N^2+13\right) t_4^2+\frac{5}{9} \left(6 N^2-5\right) t_1^3 t_5\\&+10 N \left(N^2-2\right) t_1 t_2 t_5-\frac{7}{9} \left(3 N^4-15 N^2+7\right) t_1 t_7-\frac{1}{9} N
\left(6 N^4-50 N^2+69\right) t_8
\end{aligned}&&&
\end{flalign}

\begin{flalign}
  \begin{aligned}
\mathcal{F}^5&=-\frac{8}{3} t_1^4 t_2^3-8 N t_1^2 t_2^4-\frac{16}{15} \left(3 N^2-1\right) t_2^5+\frac{4}{9} t_1^6 t_4+\frac{40}{3} N
t_1^4 t_2 t_4+16 \left(3 N^2-2\right) t_1^2 t_2^2 t_4\\ &+\frac{32}{3} N \left(2 N^2-3\right) t_2^3 t_4-\frac{16}{3} N \left(3 N^2-7\right) t_1^2 t_4^2-\frac{16}{9} \left(9 N^4-33 N^2+13\right) t_2
t_4^2\\&+\frac{40}{9} \left(6 N^2-5\right) t_1^3 t_2 t_5+40 N \left(N^2-2\right) t_1 t_2^2 t_5-\frac{20}{9} \left(9 N^4-42 N^2+20\right)
t_1 t_4 t_5\\&-\frac{5}{18} N \left(9 N^4-75 N^2+107\right) t_5^2-\frac{70}{9} N \left(N^2-3\right) t_1^3 t_7\\&-\frac{56}{9} \left(3 N^4-15 N^2+7\right) t_1
t_2 t_7-\frac{20}{9} \left(3 N^4-18 N^2+10\right) t_1^2 t_8\\&-\frac{8}{9} N \left(6 N^4-50 N^2+69\right) t_2 t_8+\frac{7}{81} \left(9 N^6-135 N^4+441
N^2-175\right) t_{10}+\frac{7}{3} N t_1^5 t_5
\end{aligned}&&&
\end{flalign}

\begin{flalign}
  \begin{aligned}
\mathcal{F}^6 &=-\frac{20}{3} t_1^4 t_2^4-16 N t_1^2 t_2^5-\frac{16}{9} \left(3 N^2-1\right) t_2^6+\frac{40}{9} t_1^6 t_2 t_4+\frac{200}{3}
N t_1^4 t_2^2 t_4
\\&+\frac{160}{3} \left(3 N^2-2\right) t_1^2 t_2^3 t_4+\frac{80}{3} N \left(2 N^2-3\right) t_2^4 t_4-\frac{4}{9} \left(69 N^2-61\right) t_1^4 t_4^2\\&-\frac{160}{3} N \left(3 N^2-7\right) t_1^2 t_2
t_4^2-\frac{80}{9} \left(9 N^4-33 N^2+13\right) t_2^2 t_4^2\\&+\frac{5}{9} t_1^7 t_5+\frac{70}{3} N
t_1^5 t_2 t_5+\frac{200}{9} \left(6N^2-5\right) t_1^3 t_2^2 t_5+\frac{400}{3} N \left(N^2-2\right) t_1 t_2^3 t_5
\\&-\frac{20}{9} N \left(39 N^2-106\right)
t_1^3 t_4 t_5-\frac{200}{9} \left(9 N^4-42 N^2+20\right) t_1 t_2 t_4 t_5\\&-\frac{25}{6} \left(9 N^4-49 N^2+27\right) t_1^2 t_5^2-\frac{25}{9} N \left(9 N^4-75
N^2+107\right) t_2 t_5^2\\&-\frac{14}{9} \left(6 N^2-7\right) t_1^5 t_7-\frac{700}{9} N \left(N^2-3\right) t_1^3 t_2 t_7-\frac{280}{9} \left(3 N^4-15 N^2+7\right) t_1
t_2^2 t_7
\\&+\frac{56}{9} N \left(6 N^4-55 N^2+89\right) t_1 t_4 t_7+\frac{35}{81} \left(18 N^6-261 N^4+855 N^2-350\right) t_5 t_7
\\&-\frac{70}{9} N \left(2
N^2-7\right) t_1^4 t_8-\frac{200}{9} \left(3N^4-18 N^2+10\right) t_1^2 t_2 t_8\\&+\frac{40}{81} \left(18 N^6-243 N^4+765 N^2-310\right) t_4 t_8+\frac{35}{9} N \left(3 N^4-35 N^2+67\right) t_1^2 t_{10}\\&+\frac{70}{81} \left(9 N^6-135 N^4+441 N^2-175\right) t_2 t_{10}-\frac{40}{9} N \left(6 N^4-50 N^2+69\right) t_2^2 t_8\\&+\frac{121}{486} \left(18 N^6-315 N^4+1197 N^2-550\right) t_1 t_{11}+\frac{16}{3} N \left(2 N^4-14 N^2+19\right) t_4^3
\end{aligned}&&&
\end{flalign}

\bibliographystyle{alphaurl}
\bibliography{Weilref}

\end{document}